\documentclass[showpacs,amsmath,amssymb,twocolumn,superscriptaddress,notitlepage,preprintnumbers,pra]{revtex4-1}

\usepackage{graphicx}
\usepackage{dcolumn}
\usepackage{bm}
\usepackage{times}
\usepackage{multirow}
\usepackage{dsfont}
\usepackage{dcolumn}
\usepackage{tabularx}
\usepackage{ae}
\usepackage{lettrine}
\usepackage{color}
\usepackage{amsmath,amsthm,amssymb,amsfonts,boxedminipage}
\usepackage{hyperref}
\usepackage{dsfont}
\usepackage{amsmath}
\newcommand{\algmargin}{\the\ALG@thistlm}
\newcommand{\alphaU}{\alpha_{\scriptscriptstyle\mathbf{U}}}

\newcommand{\alphaFP}{\alpha_{\scriptscriptstyle \mathbf{F}'}}

\makeatother

\newcommand{\norm}[1]{\left\lVert#1\right\rVert}

\newtheorem{theorem}{Theorem}
\newtheorem{lemma}{Lemma}

\newtheorem{definition}{Definition}
\newtheorem{assumption}{Assumption}
\newtheorem{remark}{Remark}

\makeatletter
\renewcommand\@biblabel[1]{#1.}
\makeatother

\usepackage[noend,noline,ruled,linesnumbered]{algorithm2e}
\usepackage{hyperref}
\usepackage{bm}
\usepackage{times}
\usepackage{multirow}
\usepackage{dcolumn}
\usepackage{tabularx}
\usepackage{graphicx}
\usepackage{ae}
\usepackage{lettrine}
\usepackage{color}
\usepackage{makecell}

\usepackage{amsmath,amsthm,amssymb,amsfonts,boxedminipage}
\usepackage{epstopdf}

\definecolor{comment}{rgb}{0, 0, 0}

\hypersetup{colorlinks=true, linkcolor=cyan, citecolor=cyan, urlcolor=blue }

\usepackage{braket}
\newcommand{\tr}[1]{\operatorname{Tr}\left(#1\right)}

\newcommand{\abs}[1]{\left| #1 \right|}

\newcommand{\PKU}{Center on Frontiers of Computing Studies, School of Computer Science, Peking University, Beijing 100871, China}
\newcommand{\bnu}{School of Artificial Intelligence,
 Beijing Normal University, Beijing,
 100875, China}

\begin{document}

\title{Heisenberg-Limited Quantum Eigenvalue Estimation for Non-normal Matrices}

\author{Yukun Zhang}
\affiliation{\PKU}

\author{Yusen Wu}
\email{yusen.wu@bnu.edu.cn}
\affiliation{\bnu}

\author{Xiao Yuan}
\email{xiaoyuan@pku.edu.cn}
\affiliation{\PKU}

\date{\today}

\begin{abstract}
Estimating the eigenvalues of non-normal matrices is a foundational problem with far-reaching implications, from modeling non-Hermitian quantum systems to analyzing complex fluid dynamics. Yet, this task remains beyond the reach of standard quantum algorithms, which are predominantly tailored for Hermitian matrices. Here we introduce a new class of quantum algorithms that directly address this challenge. The central idea is to construct eigenvalue signals through customized quantum simulation protocols and extract them using advanced classical signal-processing techniques, thereby enabling accurate and efficient eigenvalue estimation for general non-normal matrices. Crucially, when supplied with purified quantum state inputs, our algorithms attain Heisenberg-limited precision—achieving optimal performance. These results extend the powerful guided local Hamiltonian framework into the non-Hermitian regime, significantly broadening the frontier of quantum computational advantage. Our work lays the foundation for a rigorous and scalable quantum computing approach to one of the most demanding problems in linear algebra.

\end{abstract}

\maketitle

Eigenvalue estimation is a foundational computational task with widespread relevance across science and engineering. For Hermitian matrices of size \(N \times N\), quantum algorithms—such as quantum phase estimation~\cite{kitaev1995quantum}—can achieve this task with polylogarithmic complexity, offering exponential speedups over classical methods~\cite{lin2020near,lin2022heisenberg,dong2022ground}. However, many real-world problems involve \textit{non-normal matrices}, which arise naturally in contexts such as open quantum systems~\cite{breuer2002theory,de2017dynamics}, fluid dynamics~\cite{trefethen1993hydrodynamic}, control theory~\cite{hinrichsen2005mathematical}, and network analysis~\cite{asllani2018structure}. These matrices exhibit spectral characteristics that differ fundamentally from the Hermitian case: their eigenvectors are generally non-orthogonal, requiring a bi-orthogonal formalism, and their spectra can be highly sensitive to perturbations, as described by pseudospectra~\cite{trefethen1993hydrodynamic}.

Despite recent advances in quantum singular value transformation~\cite{gilyen2019quantum,lin2020near,dong2022ground}, which have extended quantum eigenvalue estimation techniques beyond Hermitian matrices, non-normal systems remain largely unexplored. Existing approaches either rely on variational quantum algorithms—generally lacking rigorous performance guarantees~\cite{zhao2023universal} —or apply to matrices with real eigenvalues~\cite{shao2022computing,low2024quantum}. Alternatively, they exploit connections between eigenvalues and singular values, which typically lead to Heisenberg-limited scaling in special cases~\cite{zhang2025exponential,lin2025quantum}, but are not capable of solving multiple eigenvalues simultaneously. More broadly, the absence of an orthogonal eigenbasis in non-normal matrices complicates both algorithmic design and theoretical analysis, positioning them as a uniquely challenging and largely untapped domain for quantum computation~\cite{horn2012matrix}.

\begin{figure}[b]
\centering
\includegraphics[width=.49\textwidth]{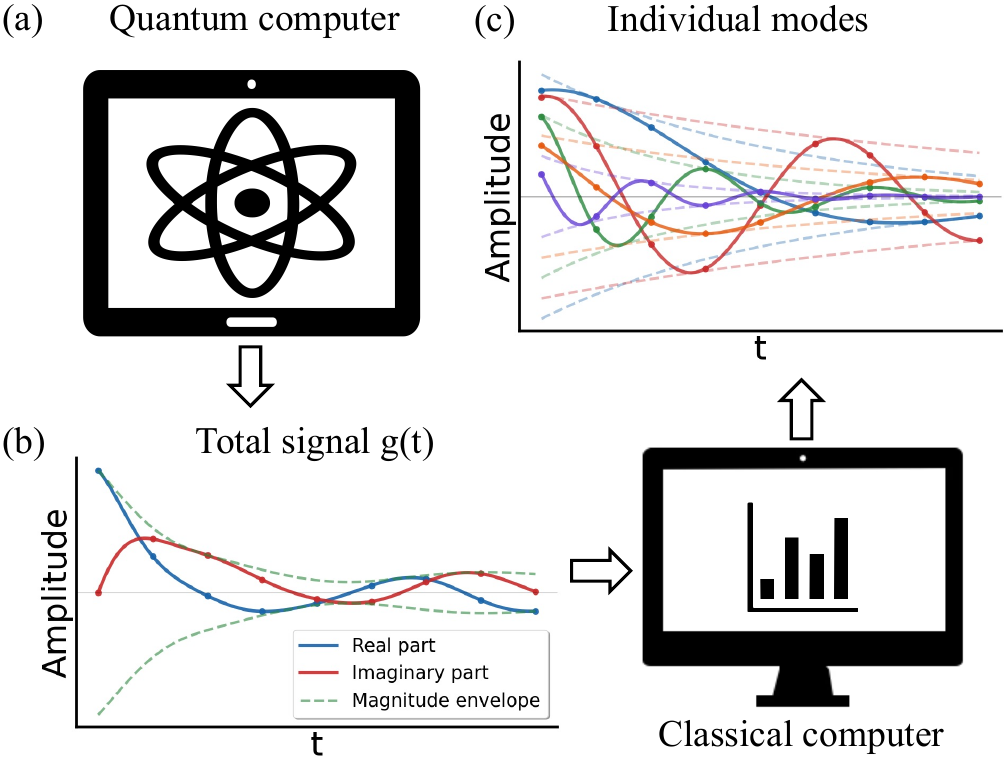} 
\caption{Schematic of the algorithm for solving the quantum eigenvalue estimation problem. 
(a) A quantum computer is used to generate times series that are mixture of signals of the eigenstates.
(b) Example of the signal of Eq.~\eqref{eq:signal1} with sparsity $r=5$. The blue and red lines are the amplitude of real and imaginary parts of the signal as a function of $t$. A classical computer is then used to solve for $\{\lambda_i\}_i$. (c) Schematic for each individual mode $\lambda_i^t$ as a function of $t$. 
Dots indicates discrete $t=1,2,\cdots,2r-1=9$. Dotted lines are the magnitude envelope of the signal, indicating exponential decay with $t$. 
}
\label{fig:main}
\end{figure}


Here, we present a systematic solution to the eigenvalue estimation problem for general non-normal matrices. As shown in Fig.~\ref{fig:main}, we first introduce tailored quantum algorithms and exploit quantum computers to construct different time-series that are a mixture of signals of different eigenvalues. These signals are then processed using efficient classical techniques, such as the matrix pencil method~\cite{hauer1990initial,hua1990matrix}, to recover the eigenvalue spectrum. Signal values are obtained via quantum measurement protocols like the Hadamard test or, more efficiently, through quantum amplitude estimation~\cite{brassard2000quantum,aaronson2020quantum} when pure states or purified quantum query access~\cite{gilyen2019quantum} is available, enabling Heisenberg-limited precision and approaching fundamental bounds on accuracy~\cite{giovannetti2004quantum,giovannetti2006quantum}. These results constitute a rigorous and general-purpose quantum framework for non-normal eigenvalue estimation, bridging a critical gap in quantum algorithm design.

In the following, we summarize the results in the main text and refer to SM for details.

\vspace{0.2cm}
\noindent\textbf{Framework ---}
Consider a general matrix \( A \in \mathbb{C}^{N \times N} \), which admits the eigen-decomposition
\begin{equation}
	A = P J P^{-1} = \sum_{i=1}^N \lambda_i \ket{\psi_i}\bra{\phi_i},
\end{equation}
where \( P \) is an invertible matrix, \( \lambda_i \) are the eigenvalues of \( A \), and \( \{ \ket{\psi_i} \} \) and \( \{ \ket{\phi_i} \} \) form a biorthogonal basis satisfying \( \langle \phi_i | \psi_j \rangle = \delta_{ij} \), with \( \delta_{ij} \) the Kronecker delta. When \( A \) is diagonalizable, \( J \) is diagonal; otherwise, \( J \) corresponds to the Jordan normal form (JNF) of \( A \).
We define the condition number of \( A \) under its JNF decomposition as $
\kappa_J(A) := \inf_{J, P} \|P\| \cdot \|P^{-1}\|$,
where the infimum is taken over all possible Jordan decompositions. Since diagonalizable matrices are dense in the space of complex matrices~\cite{horn2012matrix}, any non-diagonalizable matrix can be approximated arbitrarily closely by a diagonalizable one with small spectral perturbations. In this work, we therefore focus primarily on diagonalizable matrices.

A key challenge in extending quantum eigenvalue estimation from normal to non-normal matrices lies in the lack of an orthonormal eigenbasis. This structural difference invalidates standard techniques such as quantum phase estimation (QPE)~\cite{kitaev1995quantum,ni2023low} and filter-based algorithms~\cite{lin2020near,dong2022ground}, which rely on orthogonality. To overcome this, we formulate the quantum eigenvalue problem as follows:
\begin{definition}(\textbf{Quantum Estimation of Sparse Eigenvalues (QESE)})
Given density matrix $\rho_I$ with a sparse diagonal representation in the eigenbasis of $A$:
\begin{equation}\label{eq:dm_expansion}
\rho_I = \sum_{i=1}^r c_i |\psi_i\rangle\langle\phi_i| + \sum_{i,j:i\neq j} c_{ij} |\psi_i\rangle\langle\phi_j|,
\end{equation}
where $r \ll N$ and $\sum_{i=1}^r c_i = 1$, the task is to estimate the sparse eigenvalues $\{\lambda_i\}_{i=1}^r$.
\end{definition}

Our key observation is that any density matrix can be decomposed in the eigenbasis of the target non-normal matrix~\cite{horn2012matrix}. Under the assumption in Eq.~\eqref{eq:dm_expansion}, the diagonal components of this decomposition are assumed to be \( r \)-sparse, and the goal is to estimate the corresponding eigenvalues. This formulation naturally generalizes quantum eigenvalue estimation for Hermitian matrices~\cite{somma2019quantum}, known as the guided local Hamiltonian problem.
As we show below, the complexity of solving this quantum eigenvalue sparsity estimation (QESE) problem scales polynomially with the inverse of the minimum signal amplitude, defined as \( c_{\min} := \min_i |c_i| \). In other words, smaller spectral weights \( c_i \) in the eigenbasis expansion require proportionally more resources to resolve—mirroring the dependence on initial state overlap observed in ground-state energy estimation tasks~\cite{lin2020near,dong2022ground,lin2022heisenberg}.



Our approach consists of two steps. The first is to apply quantum circuits to generate time-series signals of the form:
\begin{equation}\label{eq:signal1}
g(t) = \tr{\rho_I A^t} = \sum_{i=1}^r c_i \lambda_i^t,
\end{equation}
where $t = 0, 1, 2, \ldots, 2r-1$. More generally, we can also construct signals using other matrix functions:
\begin{equation}\label{eq:signal2}
g'(t) = \tr{\rho_I f_t(A)} = \sum_{i=1}^r c_i f_t(\lambda_i),
\end{equation}
where $f$ is some function that is to be specified later. 
At the heart of our construction, we need to construct such matrix functions and get an estimation of the above signal through measurements.
We require block encoding (BE)~\cite{gilyen2019quantum} access to the matrix $A$, i.e., a unitary operator $U \in \mathbb{C}^{2^{m+n} \times 2^{m+n}}$ satifying
\begin{equation}\label{eq:be}
\left\| A - \alpha_A \left( \bra{G_2} \otimes I_n \right) U_A \left( \ket{G_1} \otimes I_n \right) \right\| \leq \epsilon.
\end{equation}
Here, $\alpha_A \geq \|A\|$ is the normalization factor, $m$ is the number of ancilla qubits, $\epsilon \geq 0$ is the encoding error, and $\ket{G_1}, \ket{G_2}$ are unit vectors in the ancillary space. This is called an $(\alpha_A, m, \epsilon)$-block encoding of $A\in \mathbb{C}^{2^{n} \times 2^{n}}$. In addition, we need to design adapted quantum eigenvalue value transformations (QEVT) to realize the matrix functions following the framework firstly proposed in Ref.~\cite{low2024quantum}. 

The second step involves extracting the eigenvalues from the generated time-series signals. To this end, we employ the matrix pencil (MP) method~\cite{hua1990matrix}, and provide a detailed analysis of error propagation within the MP framework. This allows us to derive noise-tolerant bounds for solving the QESE problem reliably under realistic conditions.

While the core techniques leveraged in this work—such as QEVT and MP—have been individually studied in other contexts, our approach is the first to integrate them into a unified and systematic framework for non-normal eigenvalue estimation. The resulting algorithms offer a significantly broader and more versatile solution to this long-standing challenge. Meanwhile, from a technical perspective, our contributions also include the development of problem-specific QEVT techniques and a thorough analysis of the algorithm’s complexity and robustness to noise.


\vspace{0.2cm}
\noindent\textbf{Quantum signal generation ---}
We first introduce the procedure for generating signals of Eq.~\eqref{eq:signal1} and \eqref{eq:signal2}. 

Given the $(\alpha_A, m, 0)$-BE $U_A$ of $A$, we can construct the $(\alpha_A^t, mt, 0)$-BE $U_{A^t}$ of $A^t$ by taking product of the block encoded matrix~\cite{gilyen2019quantum}. The signal generation process then proceeds by implementing the Hadamard test circuit with the controlled unitary being $U_{A^t}$ as shown in Fig.~\ref{fig:main} (a). Setting $W=I$ or $S^\dagger$ with $S$ the phase gate, the measurement outcomes of the ancillary qubit are serve as unbiased estimators to $\mathrm{Re}(g(t)/\alpha_A^t)$ and $\mathrm{Im}(g(t)/\alpha_A^t)$, where $\mathrm{Re}(\cdot)$ and $\mathrm{Im}(\cdot)$ denote taking the real and imaginary parts. The final step of generating the signal in Eq.~\eqref{eq:signal1} follows multiplying the factor $\alpha_A^t$ to the estimated results from the Hadamard test circuit. The sampling-shot-to-accuracy rate of the Hadmard test protocol is dictated by Hoeffding's inequality~\cite{hoeffding1994probability}, such that $m$ shots of measurements yield a $1/\sqrt{m}$ accuracy with high probability. When we have purified quantum query access to $\rho_I$, we can quadratically improve the accuracy to $1/m$ using the QAE~\cite{brassard2000quantum,aaronson2020quantum} algorithm, while increasing the circuit depth in one coherent implementation. 


In the special cases where the eigenvalues of the matrix $A$ are real, we may significantly improve the performance of the algorithm by casting the signal to the form of a linear combination of Fourier modes (exponentials). This reduces the signal processing procedure to a super-resolution problem. To this end, we resort to the QEVT framework \cite{low2024quantum} for implementing the function $x\mapsto e^{-2\pi i xt}$ so that the eigenvalues are mapped to a phase factor. Conceptually, the QEVT algorithm (approximately) performs a matrix function $A\mapsto f(A)=Pf(\Lambda)P^{-1}$ with $P\Lambda P^{-1}$ being the eigendecomposition of $A$. As such, by approximating the above-stated mapping to the blocked encoded matrix, we can obtain a $(\alpha_p, k, \varepsilon)$-BE to $e^{-2\pi i A t/\alpha_A}$ with some $k$ and $\varepsilon$. Crucially, the normalization factor $\alpha_p$ can be chosen to be of order $\mathcal{O}(1)$, which we can readily multiply back to get an approximation to the ideal signal. The ideal signal thus becomes
\begin{equation}\label{eq:signal3}
g'(t) = \tr{\rho_I e^{-2\pi i A t/\alpha_A}} = \sum_{i=1}^r c_i e^{-2\pi i \lambda_i t/\alpha_A}.
\end{equation}
We remark that the signal shares similarity with the recently intensively studied phase estimation problems with early fault-tolerant devices~\cite{lin2022heisenberg,ni2023low,ding2023even}.
Using either the Hadamard test or QAE, we obtain an estimation of the ideal signal. 



\vspace{0.2cm}
\noindent\textbf{Matrix Pencil Method ---}
Once we obtain each estimated signal, we apply the matrix pencil method~\cite{hua1990matrix} to extract eigenvalues. This method leverages the shift-invariance property of signals composed of powers or exponentials and works as follows: 1.~Construct Hankel matrices $H_0 = (g(j+k))_{j,k=0}^{r-1}$ and $H_1 = (g(j+k+1))_{j,k=0}^{r-1}$.
2.~Define the generalized eigenvalue problem (GEVP): 
    $H_1 |{\nu}\rangle= z H_0 |{\nu}\rangle$.
3.~Each generalized eigenvalue $z$ correspond to each component in the signal with $t=1$. These are $\lambda_i$ and $e^{-2\pi i\lambda_i/\alpha_A}$ for Eq.~\eqref{eq:signal1} and \eqref{eq:signal3} , respectively.
4.~Post-process the obtained generalized eigenvalues to get an estimation of each eigenvalue in the QESE problem.
The key insight is that these Hankel matrices have a Vandermonde decomposition:
$H_0 = V(\mathbf{z})CV(\mathbf{z})^T$ and
   $H_1 = V(\mathbf{z})CZV(\mathbf{z})^T$,
where $V(\mathbf{z}) = (z_j^{k-1})_{j,k=1}^r$ is the Vandermonde matrix with $z_j$ the generalized eigenvalue, $C = \text{diag}(\mathbf{c})$ is the diagonal matrix of coefficients $c_i$, and $Z = \text{diag}(\mathbf{z})$ contains the (transformed) eigenvalues we seek. It is readily verified that the GEVP is given by the diagonal elements of $Z$ as claimed.
In practics, while $r$ is unknown, we assume a known upper bound $R\ge r$ and hence first construct the $R\times R$ Hankel matrix and estimate its rank to learn about the sparsity.



The next significant question is that \emph{how robust is the MP method to the noises in the Hankel matrices?} The importance lies in that the estimated signals inevitably suffer from two sources of errors, the approximation error in the matrix function and estimation inaccuracy brought by the Hadamard test or QAE, as discussed in the last section. High sensitivity to noise or inaccuracy could render the algorithm inefficient for solving the QESE problem. To this end, we establish the fault-tolerant bounds of the MP method. We briefly discuss how the error propagates in the MP method and leave the rigorous analysis to Supplementary Information. The GEVP can be reduced to an eigenvalue problem when the Vandermonde matrix is invertible:
$H^{-1}_0H_1\ket{\nu}=\lambda\ket{\nu}$.
Denoting $G:=H^{-1}_0H_1$, we find its eigendecomposition as $G=(V^{\rm T})^{-1}ZV^{\rm T}$. First order perturbation theory shows that small perturbation (noises) in the Hankel matrix is enlarged by a factor proportional to $\|G\|\leq \|(V^{\rm T})^{-1}\|\|Z\|\|V^{\rm T}\|\leq\kappa(V)$, where $\|x\|$ denotes the operator norm of $x$, and $\kappa(V)$ denotes the condition number of $V$. Besides, we have used the sub-multiplicativity of the operator norm and $\|Z\|\leq 1$. Established results have shown that there is a strong separation in the behavior of the condition number, depending on whether the components in the signal are Fourier modes or not. In the former cases, the condition number could be well-behaved (grow at most polynomially with $r$)~\cite{moitra2015super}. By contrast, in the latter case, a constant $r=\mathcal{O}(1)$ is sufficient for guaranteeing the noise robustness. 


\vspace{0.2cm}
\noindent\textbf{Complexity analysis ---} 
Our algorithm applies generally to all non-normal matrices with complex or real eigenvalues. Here, we focus on the results with purified quantum query access to $\rho_I$ and refer to Supplementary Information for results with sample access to $\rho_I$. 
The result for general complex eigenvalues is:
\begin{theorem}[informal]
For the targeted eigenvalues $Z:=\{\lambda_i\}_{i=1}^r$ of $A\in\mathbb{C}^{N\times N}$, there exists a quantum algorithm that outputs an estimation $\{\widetilde\lambda_i\}_{i=1}^r$ such that
$\min _{\lambda \in Z}|\widetilde{\lambda}-\lambda|\leq \epsilon$
with high success probability and $\mathcal{\widetilde{O}}\left(\exp (r)  \epsilon^{-1}  c_{\min }^{-1} \alpha_A^{q_1 r}\left(\Delta^{\prime}\right)^{-q_2r}\right)$ queries to (controlled) $U_A$ and its inverse in one coherent run of the algorithm, and $\mathcal{{O}}\left(r\right)$ queries to $\rho_I$. Here $q_1,q_2>0$ are constant numbers.
\end{theorem}
\noindent The algorithm remains efficient when $r=\mathcal{O}(1)$, which is often the case in practical applications like identifying the slowest decay modes in open quantum systems. Besides, we note that the performance of the algorithm depends inversely on $c_{\min}$ and $\Delta'$. The first one follows that, as the eigen-component in the initial state becomes smaller, it becomes harder to detect it. Besides, the normalized gap $\Delta'$ decides the hardness in distinguishing different generalized eigenvalues used to solve for the target eigenvalues. This necessitates the separation of different eigenvalues to not be exponentially small, a requirement also found in the classical signal processing studies~\cite{moitra2015super}.

The purified quantum query access to $\rho_I$ enables to achieve the Heisenberg limit performance, which is believed to be optimal. While achieving better total asymptotic scaling, the algorithm based on QAE requires far deeper circuit implementation in one coherent run of the algorithm, which could hinder its applicability for early fault-tolerant quantum devices. Note that, although this method cannot detect the eigenvalue $0$ as the corresponding signal vanishes, in Supplementary Information, we construct a different method that resorts to QEVT that solves the problem.

When the eigenvalues are real, the algorithm can be improved and works even for polynomially large $r$. 
\begin{theorem}[informal]
Suppose the eigenvalues of $A\in\mathbb{C}^{N\times N}$ are real, the cost is improved to $\mathcal{\widetilde{O}}\left({\alpha_A \epsilon^{-1}c_{\min}^{-1}\kappa^4(V)d^{3/2}\kappa_J^2}{}\right)$ queries to controlled-$U_A$ and its inverse in one coherent run of the algorithm, and $\mathcal{O}\left(r\right)$ queries to $\rho_I$, where $\kappa^2(V) \leq \frac{r+1 / \Delta_w-1}{r-1 / \Delta_w-1}$.

\end{theorem}
\noindent The upper bound on the condition number of $V$ is proved in the seminal work of Ref.~\cite{moitra2015super}, such that when the signal is well-separated ($\Delta_w$ is not too small), the condition number remains reasonable. As such, the algorithm could remain efficient even when $r=\mathcal{O}(\mathrm{poly}(n))$.



\vspace{0.2cm}
\noindent\textbf{Applications ---}
Our algorithms find natural applications in domains where non-normal matrices appear, such as non-hermitian physics, open quantum systems, fluid dynamics, differential equations, etc. Here, we discuss two application in estimating the Liouvillian gap for open quantum systems and the spectral abscissa and dynamics stability in fluid dynamics. 


In open quantum systems, the Liouvillian superoperator $\mathcal{L}$ governs dynamics through the master equation $\frac{d\rho}{dt} = \mathcal{L}(\rho)$. The Liouvillian gap—the smallest non-zero value of the real part of eigenvalues—determines relaxation timescales. Our algorithm efficiently estimates this gap by vectorizing the Liouvillian into a non-normal matrix and applying our complex eigenvalue algorithm, providing crucial insights into dissipative quantum dynamics. For certain systems that possess PT symmetry and have real eigenvalues~\cite{bender1998real}, we show in Supplementary Information improved algorithm for the enhanced dependence on the sparsity $r$.

The spectral abscissa of a matrix is a fundamental quantity in the stability analysis of dynamical systems, particularly in fields like control theory, fluid dynamics, and differential equations.
In fluid dynamics, linearizing the Navier-Stokes equations yields non-normal operators $\mathcal{N}$.
Essentially, the spectral abscissa $a(\mathcal{N}):=\max\{\mathrm{Re}(\lambda):\lambda\in\sigma(\mathcal{N})\}$ with $\sigma(\mathcal{N})$ denotes the set of eigenvalues, determine the stability of the dynamics.
Our algorithm for general complex eigenvalues enables efficient analysis of these operators and the spectral abscissa.
In special cases where the matrices are known to have real eigenvalues, such as oscillatory matrices in mechanical vibration analysis~\cite{gantmacher2002oscillation}, we can adopt the algorithm tailored for real eigenvalues to improve the performance.

\vspace{0.2cm}
\noindent\textbf{Conclusion ---}
In summary, this work introduces quantum algorithms designed to address the quantum eigenvalue problem for non-normal matrices. Assuming access to an initial state whose representation in the target matrix's eigenbasis is diagonally sparse, we presented distinct algorithmic frameworks tailored for general complex, purely real, and other complex eigenvalue scenarios. Our approach leverages the synergy between QEVT and classical signal processing, specifically the matrix pencil method. Quantum circuits generate characteristic signals encoding eigenvalue information, which are then processed classically to extract the desired spectral data.

Our complexity analysis establishes conditions for algorithmic efficiency, primarily dependent on eigenvalue sparsity ($r$), minimum diagonal coefficient ($c_{\min}$), and spectral gaps ($\Delta'$ or $\Delta_w$). Notably, for general complex eigenvalues, efficiency typically requires constant sparsity ($r=\mathcal{O}(1)$), whereas for real eigenvalues, polynomial sparsity ($r=\mathcal{O}(\text{poly}(n))$) can be accommodated under favorable gap conditions. A key finding is the achievement of near-Heisenberg-limited precision scaling ($\mathcal{O}(\epsilon^{-1})$) when purified quantum query access is available.

Several promising avenues for future research emerge from this work. Extending these techniques beyond eigenvalue estimation to the preparation of corresponding non-orthogonal eigenstates presents an intriguing challenge, potentially requiring novel filtering techniques compatible with the complexities of non-normal operators, particularly in the complex domain where polynomial functions may not be bounded as opposed to the normal matrices cases. Investigating methods to further optimize the computational complexity, especially the dependence on sparsity and spectral gaps, remains an important objective. Establishing rigorous lower bounds for the QESE problem and designing algorithms that provably achieve these bounds would provide deeper insights into the ultimate capabilities and limitations of quantum computation for non-normal eigenvalue problems.

\section*{Acknowledgments}
\noindent We acknowledge the support from the National Natural Science Foundation of China  NSAF Grant (No.~U2330201) and Grant No.~12361161602, and the Innovation Program for Quantum Science and Technology Grant (No.~2023ZD0300200).
Y.~Wu acknowledges the support from the Natural Science Foundation of China Grant No.~62371050. 
\bibliography{ref}
\clearpage

\widetext
\section*{Supplementary Information}
\appendix

\section{Preliminaries}
\subsection{Basics of matrix decompositions and transformations}
\label{sec:basic_matrix}
In this work, we consider $n$-qubit square matrix $A\in\mathbb{C}^{N\times N}$ with $N=2^n$. 
The eigenvalues of an arbitrary square matrix are characterized by solving 
\begin{equation}\label{eq:character_poly}
    \det[A-\lambda I]=\sum_{\sigma \in S_N} \operatorname{sgn}(\sigma) \prod_{i=1}^n A'_{\sigma(i) i}(\lambda)=0,
\end{equation}
where $\det[X]$ represents determinant of $X$, $I$ is the identity matrix, $A'(\lambda)=A-\lambda I$, $\operatorname{sgn}$ is the sign function and $\sigma$ is an element from the permutation group $S_N$. Here, we have expanded the determinant into a polynomial function, known as \emph{characteristic polynomial}, of $\lambda$ using the Leibniz formula. As such, the eigenvalues are solved as the roots of the polynomial function. Any matrix $A$ can be decomposed as the sum of its hermitian and anti-hermitian parts:
\begin{equation}\label{eq:ha_decomp}
    A=A_h+iA_a,
\end{equation}
such that 
$$A_h=\frac{A+A^\dagger}{2},~A_a=\frac{A-A^\dagger}{2i}$$
satisfying $A_h^\dagger=A_h$ and $(iA_a)^\dagger=-iA_a$.

A matrix is \emph{normal} if it commutes with its complex conjugation such that $[A,A^\dagger]=0$. In this case, the eigendecomposition of the matrix is given by $A=UDU^\dagger$, where $U$ is unitary and $D={\rm{diag}}(\lambda_1,\cdots,\lambda_N)$ with $\lambda_i$ the eigenvalues of $A$. 
In such cases, the matrix can be diagonalized by a set of orthonormal bases $\{\ket{\psi_i}\}_{i=1}^N$ such that
$$A\ket{\psi_i}=\lambda_i \ket{{\psi_i}},$$
where each $\ket{{\psi_i}}$ is the column vector of $U$ satisfying $\langle \psi_i \ket{\psi_j}=\delta_{ij}$ with $\delta_{ij}$ the Kronecker delta function and we have $A=\sum_i \lambda \ket{\psi_i}\bra{\psi_i}$. 
It is apparent that the orthonormal basis is complete, such that $$\sum_{i=1}^N\ket{\psi_i}\bra{\psi_i}= I,$$
where $ I$ is identity matrix of size $N\times N$.
Typical examples of normal matrices include hermitian, anti-hermitian, and unitary matrices. 

For matrices such that $[A,A^\dagger]\neq0$, these are \emph{non-normal} matrices. For those which are \emph{diagonalizable}, we have the eigendecomposition to be $A=P\Lambda P^{-1}$, where $P$ is an invertible matrix and 
$\Lambda=\rm{diag}(\lambda_1,\cdots,\lambda_N)$ with $\lambda_i\in\mathbb{C}$ the eigenvalues of $A$. It is known that a matrix is diagonalizable, also known as \emph{non-defective}, if and only if the \emph{geometric multiplicity} equals the \emph{algebraic multiplicity} for every eigenvalue. Formally, the algebraic multiplicity of an eigenvalue is the number of times it appears as a root in the characteristic polynomial of the matrix given by Eq.~\eqref{eq:character_poly}. The geometric multiplicity of an eigenvalue is the dimension of the eigenspace corresponding to that eigenvalue. In other words, it is the number of linearly independent eigenvectors associated with the eigenvalue.

Unlike normal matrices, diagonalizable non-normal matrices are diagonalized by two sets of bi-orthogonal bases. Denoted $\{\ket{\psi_i}\}_{i=1}^N$ and $\{\ket{\phi_i}\}_{i=1}^N$ be to the column and row vectors of $P$ and $P^{-1}$,  respectively, the two set of bases are biorthogonal because $\langle \psi_i\ket{\phi_j}=\delta_{ij}$. One finds that
\begin{equation}\label{eq:spectral_rep}
    A=\sum_{i=1}^N\lambda_i\ket{\psi_i}\bra{\phi_i}.
\end{equation}
The biorthogonal basis is complete due to
$$\sum_{i=1}^N\ket{\psi_i}\bra{\phi_i}= I.$$
For a mathematical proof, we refer the readers to Ref.~\cite{horn2012matrix,gohberg2005matrix}. For the matrix function, we want to implement $f(A)$ for some function $f(x)$. Because of the biorthogonality, it is easy to verify that 
\begin{equation}
    f(A)=\sum_{i=1}^Nf(\lambda_i)\ket{\psi_i}\bra{\phi_i}=Pf(\Lambda)P^{-1}=P
    \begin{bmatrix}
        f(\lambda_1) & & &\\
        & f(\lambda_2) & &\\
        & & \ddots &\\
        & & & f(\lambda_{N})
    \end{bmatrix}P^{-1}.
\end{equation}

In more general cases, a non-normal matrix may not be diagonalizable, i.e.~the matrix is \emph{defective}. This happens when at least one eigenvalue has a geometric multiplicity that is less than its algebraic multiplicity. That is, the number of linearly independent eigenvectors for an eigenvalue is less than the number of times that eigenvalue is repeated. Under such circumstances, the eigendecomposition is given by \emph{Jordan form decomposition} (JFD) $A=PJP^{-1}$ with $P$ some invertible matrix and $J$ the \emph{Jordan canonical form} (JCF) of $A$ given by
\begin{equation}
        \label{eq:jordan_block}
        \begin{aligned}
            J&=
            \begin{bmatrix}
                J(\lambda_1,d_1) & & &\\
                & J(\lambda_2,d_2) & &\\
                & &\ddots &\\
                & & &J(\lambda_{M},d_{M})
            \end{bmatrix}.
        \end{aligned}
        \end{equation}
Here, each $J(\lambda_l,d_l)\in\mathbb{C}^{d_l\times d_l}$ is a \emph{Jordan block} of block diagonal form
\begin{align*}
    J(\lambda_l,d_l)&=
            \begin{bmatrix}
                \lambda_l &  &  &  &  & \\
                1 & \lambda_l &  &  &  & \\
                0 & 1 & \lambda_l &  &  & \\
                \vdots & 0 & 1 & \ddots &  & \\
                \vdots & \ddots & \ddots & \ddots & \lambda_l & \\
                0 & \cdots & \cdots & 0 & 1 & \lambda_l\\
            \end{bmatrix}.
\end{align*}
When all block dimensions of the Jordan blocks equal one, a matrix is diagonalizable. Unlike non-defective matrices, the matrix function of defective matrices is given by \cite[Sec.~2.3]{low2024quantum}:
\begin{equation}
\label{eq:jordan_form_transformation}
\begin{aligned}
    f(A)&=Pf(J)P^{-1}
    =P
    \begin{bmatrix}
        f(J(\lambda_0,d_0)) & & &\\
        & f(J(\lambda_1,d_1)) & &\\
        & &\ddots &\\
        & & &f(J(\lambda_{s-1},d_{s-1}))     
    \end{bmatrix}P^{-1},\\
    f(J(\lambda_l,d_l))&=
    \begin{bmatrix}
        f(\lambda_l) &  &  &  &  & \\
        f'(\lambda_l) & f(\lambda_l) &  &  &  & \\
        \frac{f^{(2)}(\lambda_l)}{2!} & f'(\lambda_l) & f(\lambda_l) &  &  & \\
        \vdots & \frac{f^{(2)}(\lambda_l)}{2!} & f'(\lambda_l) & \ddots &\\
        \vdots & \ddots & \ddots & \ddots & \ddots & \\
        \frac{f^{(d_l-1)}(\lambda_l)}{(d_l-1)!} & \cdots & \cdots & \frac{f^{(2)}(\lambda_l)}{2!} & f'(\lambda_l) & f(\lambda_l).\\
    \end{bmatrix}
\end{aligned}
\end{equation}
This shows that the matrix functions of defective matrices rely on higher derivatives. A critical quantity is the maximal size of the Jordan block in $J$, which is formally defined as
\begin{equation}
    d_{\max }:=\max_\lambda \left\{\min \left\{k \in \mathbb{N} \mid \operatorname{ker}\left((A-\lambda  I)^k\right)=\operatorname{ker}\left((A-\lambda  I)^{k+1}\right)\right\}\right\},
\end{equation}
where $\operatorname{ker}(\cdot)$ denotes the kernel. This factor relates directly to the complication of matrix functions~\cite[Sec.~1.2.2]{higham2008functions}.

Such a form of function complicates our analysis. Yet, we resort to the fact that diagonalizable matrices with simple spectrum (distinctive eigenvalues) are dense in the space of matrices~\cite{horn2012matrix}. This indicates that for a defective matrix $A$, we can find a non-defective one $A'$ by perturbing the entries of $A$. For a sufficiently small perturbation, the eigenvalues of $A'$ are close to those of $A$. Thus, in this work, we only consider non-defective matrices. For diagonalizable matrices, we find that $d_{\max}=1$. Besides, we also define the \emph{Jordan condition number} as 
\begin{equation}
    \kappa_J(A):=\inf_{J,P}\norm{P}\norm{P^{-1}}
\end{equation}
with minimization over all possible forms of JNF. In cases the matrix becomes diagonalizable, this factor gauges how much the matrix deviates from a normal one~\cite{trefethen2020spectra}. 


\subsection{Block encoding and other quantum algorithmic tools}
\label{sec:be}
Block encoding is a foundational technique in quantum algorithms that enables the embedding of matrices into unitary operators, facilitating a broad range of quantum operations. We present a definition as follows:
\begin{definition}[Block Encoding]
\label{def:block_encoding}
Let $A \in \mathbb{C}^{2^n \times 2^n}$ be a matrix. A unitary operator $U \in \mathbb{C}^{2^{m+n} \times 2^{m+n}}$ is called an $(\alpha_A, m, \epsilon)$-block encoding of $A$ if
\begin{equation}\label{eq:be}
\left\| A - \alpha_A \left( \bra{G_2} \otimes I_n \right) U_A \left( \ket{G_1} \otimes I_n \right) \right\| \leq \epsilon,
\end{equation}
where $\alpha_A \geq \|A\|$ is the normalization factor, $m$ is the number of ancilla qubits, $\epsilon \geq 0$ is the encoding error, and $\ket{G_1}, \ket{G_2}$ are unit vectors in the ancillary space.
\end{definition}
The block encoding framework provides a probabilistic implementation of matrix operations. When measuring the ancillary register in state $\ket{G_2}$ after applying $U$ to $\ket{G_1} \otimes \ket{\psi}$, one obtains a state proportional to $A\ket{\psi}$ with probability $\|A\ket{\psi}\|_2^2/\alpha_A^2$, where $\|\cdot\|_2$ denotes the $\rm{L}_2$ norm. This probabilistic nature is fundamental to quantum algorithms that leverage amplitude amplification~\cite{yoder2014fixed} to enhance success probabilities.

The block encoding framework is flexible and allows several basic operations, such as a linear combination of several block-encoded matrices and matrix multiplication. In particular, we provide the following lemma for the realization of products of block-encoded matrices for later convenience.
\begin{lemma}[Product of multiple block-encoded matrices {\cite[Adapted from Lemma 53]{gilyen2019quantum}}]\label{lemma:block_prod}
Let $U_A$ be a $(\alpha_A, m, 0)$-block encoding of $A$ with $\ket{G_1}=\ket{G_2}=\ket{0^m}$ as defined in Definition \ref{def:block_encoding}. Then, there exists $U_{A^t}$, which is a $(\alpha_A^t, mt, 0)$-block encoding of $A^t$.
\end{lemma}
\begin{proof}
This follows directly from iteratively applying Lemma 53 of Ref.~\cite{gilyen2019quantum} $(t-1)$ times.
\end{proof}


\begin{lemma}[Block encoding amplification {\cite[Theorem 2]{low2017hamiltonian}}]
\label{lem:amp_block}
    Let $C$ be a matrix such that $C/\alpha$ is block encoded by $O_C$ with some normalization factor $\alpha$.
    Then given any $\alpha_C\geq\norm{C}$, the operator
    \begin{equation}
        \frac{C}{2\alpha_C}
    \end{equation}
    can be block encoded with accuracy $\epsilon$ using
    \begin{equation}
        \mathcal{O}\left(\frac{\alpha}{\alpha_C}\log\left(\frac{1}{\epsilon}\right)\right)
    \end{equation}
    queries to the controlled-$O_C$ and its inverse.
\end{lemma}
We also borrow the block encoding version of the QLSP algorithm. That is, block encodes a normalized inverse matrix.
\begin{lemma}[Block encoding inversion {\cite[Corollary 69]{gilyen2019quantum}}]
\label{lem:inv_block}
    Let $\epsilon\in(0,\frac{1}{2}]$. Let $O_X$ be a ($\alpha_X, m,0$)-block encoding of $X$ with $\alpha_X\geq\norm{X}$. Then, there is a ($\alpha_{X^{-1}}, m+1, \epsilon$)-block encoding of $X^{-1}$ using
     $$
    \mathcal{O}\left(\alpha_X\alpha_{X^{-1}}\log\left(\frac{1}{\epsilon}\right)\right)
    $$
    queries to the controlled-$O_X$ and its inverse such that $\alpha_{X^{-1}}\geq\norm{X^{-1}}$ upper bounds the norm of the inverse matrix.
\end{lemma}

Along this side, we also employ the quantum amplitude estimation (QAE) method to facilitate later discussion.

\begin{lemma}[Quantum amplitude estimation {\cite[Adapted from Theorem 3]{aaronson2020quantum}}]\label{lemma:ae}
Suppose $U$ is an $(a+b)$-qubit unitary operator such that
$$
U|0\rangle_{a+b}=\sqrt{p}|0\rangle_a\left|\phi_0\right\rangle_b+\sqrt{1-p}|1\rangle_a\left|\phi_1\right\rangle_b,
$$
where $\left|\phi_0\right\rangle$ and $\left|\phi_1\right\rangle$ are pure quantum states and $p \in[0,1]$. There exists a quantum algorithm that outputs $\widetilde{p} \in[0,1]$ such that
$$
|\widetilde{p}-p| \leq \epsilon
$$
with probability at least $(1-\delta)$ for $\delta\in[0,1]$, using $\mathcal{O}(\epsilon^{-1}\log(\delta^{-1}))$ queries to $U$ and $U^{\dagger}$ for sufficiently small $\epsilon$.
\end{lemma}
\begin{proof}
We invoke \cite[Theorem 3]{aaronson2020quantum}. There, with probability no less than $1-\delta$, a query complexity of $\mathcal{O}(\varepsilon_1^{-1}{p}^{-1/2}\log(\delta^{-1}))$ of $U$ and $U^\dagger$ is applied to output an estimation $\hat{x}$ of $\sqrt{p}$ such that
$$(1-\varepsilon_1)\sqrt{p}<\hat{x}<(1+\varepsilon_1)\sqrt{p}.$$
This achieves relative error $\varepsilon_1$ in estimating $\sqrt{p}$. In absolute error, we set $\varepsilon_1=\varepsilon/\sqrt{p}$ so that
$$|\hat{x}-\sqrt{p}|<\varepsilon.$$
Thus, let $\hat{x} = \sqrt{p} + \eta,~|\eta|<\varepsilon$. By setting $\hat{x}^2$ as an estimation to $p$, we find
\begin{align*}
  \bigl|\hat{x}^2-p\bigr|
  = \bigl|2\sqrt{p}\,\eta + \eta^2\bigr|
  \;\le\; 2\sqrt{p}\,|\eta| + \eta^2
  \;\le\; 2\sqrt{p}\,\varepsilon + \varepsilon^2
  \;\le\; 2\,\varepsilon + \varepsilon^2.
\end{align*}
For sufficient small $\epsilon$, we may set $\varepsilon=\epsilon/2$. This results in the claimed complexity.
\end{proof}

Compared to the seminal work of Ref.~\cite{brassard2000quantum}, the above results simplified the circuit implementation such that it does not require a quantum Fourier transform, which could render the algorithm more amenable to early fault-tolerant quantum devices.

\subsection{Chebyshev and Faber polynomials}\label{sec:cheby_faber}
We introduce the tools in approximation theory~\cite{sachdeva2014faster,trefethen2019approximation} for approximating a target function using polynomials. The theory aims to provide the (near) optimal uniform approximation.
That is, given a function $f \in C[a,b]$ and an approximant $g$ from a prescribed class, one seeks to minimize the supremum norm:
\begin{equation}
\|f-g\|_{\infty} = \sup_{x \in [a,b]} |f(x) - g(x)|.
\end{equation}
The best polynomial approximation is then given by the lowest degree polynomial function with desirable accuracy in the uniform approximation. The constructive aspect of optimal uniform approximation is obtained by the Chebyshev polynomials of the first kind, denoted as $T_d(x)$. The function provides (near) optimal performance for functions defined on the canonical interval $[-1,1]$, and can be defined via the trigonometric relation:
$$
T_d(x) = \cos(d\arccos(x)).
$$
To see that it is indeed a polynomial function, one finds that they can also be formulated
through the recursive formulation:
\begin{equation*}
\begin{aligned}
T_0(x) &= 1\\
T_1(x) &= x\\
T_{d+1}(x) &= 2xT_d(x) - T_{d-1}(x).
\end{aligned}
\end{equation*}
The Chebyshev polynomials manifest an equioscillatory behavior characterized by attaining their extremal values of precisely $\pm 1$ at $n+1$ distinct points within $[-1,1]$ and possessing exactly $n$ zeros interspersed among these extrema. This underpins them to be the best uniform polynomial approximants as dictated by the Chebyshev equioscillation theorem~\cite{trefethen2019approximation}. Following Ref.~\cite{low2024quantum}, we define
\begin{equation}
    \widetilde{{T}}_j(x)=\begin{cases}{T}_j(x),&j\geq1,\\\frac{1}{2}{T}_0(x),&j=0\end{cases}
\end{equation}
such that the first Chebyshev polynomial is rescaled by a factor of $\frac{1}{2}$. Then, consider a Chebyshev expansion $f(x)=\sum_{j=0}^{\infty}\beta_j{T}_j(x)=\sum_{j=0}^{\infty}\widetilde{\beta}_j\widetilde{T}_j(x)$ with $\beta_j$ the coefficients of the expansion, and $\widetilde{\beta_j}$ the rescaled coefficient:
\begin{equation}
    \widetilde{\beta}_j=\begin{cases}\beta_j,&j\geq1,\\2\beta_0,&j=0.\end{cases}
\end{equation}

When extending the uniform approximation theory to the complex domain, Faber polynomials~\cite{markushevich2005theory} emerge as the preeminent candidate for approximation on compact subsets of the complex plane, exhibiting optimality properties analogous to their Chebyshevian counterparts. For a compact set $\mathcal{E} \subset \mathbb{C}$ with connected complement, let $\Omega = \mathbb{C}_{\infty} \setminus \mathcal{E}$ denote the complement of $\mathcal{E}$ with respect to the extended complex plane.
Here, we choose $\mathcal{E}$ as a region that encloses the eigenvalues of the target matrix $A$. We will refer to this region as \emph{Fiber region}.

Under the assumption that $\Omega$ is simply connected, there exists a unique conformal mapping $\Psi: \Omega \rightarrow \Delta = \{w \in \mathbb{C} : |w| > 1\}$ with the normalization $\Psi(\infty) = \infty$ and $\lim_{z \rightarrow \infty} \Psi(z)/z > 0$. This map is known as the \emph{exterior Riemann map}.
The Faber polynomials $\{F_n(z)\}_{n=0}^{\infty}$ associated with $\mathcal{E}$ are then defined via the generating function:
$$
\frac{\Psi'(z)}{\Psi(z) - w} = \sum_{n=0}^{\infty} \frac{F_n(z)}{w^{n+1}}, \quad |w| > \sigma > \|\Phi\|_\mathcal{E}
$$
where $\sigma$ denotes the capacity of $\mathcal{E}$ and $\|\Psi\|_\mathcal{E} = \sup_{z \in \mathcal{E}} |\Psi(z)|$.

Importantly, the Faber polynomials constitute a natural generalization of the monomial basis to arbitrary compact sets, reducing to the standard basis $\{z^n\}_{n=0}^{\infty}$ when $\mathcal{E}$ is the unit disk. That is when $\mathcal{E}=\{|w|\leq 1\}$, the exterior Riemann map is the identity map $\Psi(w)=w$ and 
$$F_n(z)=z^n$$
is the power function.
More significantly, for the canonical interval $[-1,1]$, they coincide with the Chebyshev polynomials of the first kind, thereby establishing an elegant theoretical continuity between real and complex approximation theory.

The theoretical significance of Faber polynomials is encapsulated in their near-optimal properties for uniform approximation in the complex domain. Specifically, for a function $f$ analytic on $E$ and analytic in some neighborhood of $E$, the truncated Faber series:
$$
S_n(f,z) = \sum_{k=0}^{n} a_k F_k(z),
$$
where the coefficients $a_k$ are determined by the asymptotic expansion of $f$ at infinity:
$$
f(z) \sim \sum_{k=0}^{\infty} \frac{a_k}{[\Phi(z)]^k}, \quad z \rightarrow \infty
$$
provides a near-optimal polynomial approximation with respect to the uniform norm on $E$.

The approximation error admits the precise characterization:
$$
\|f - S_n(f)\|_E \leq C \cdot \inf_{p \in \mathcal{P}_n} \|f - p\|_E,
$$
where $\mathcal{P}_n$ denotes the space of polynomials of degree at most $n$, and $C$ is a constant independent of $n$ but dependent on the geometric properties of $E$. For certain classes of domains, particularly those with analytic boundaries, this constant can be shown to be asymptotically optimal as $n \rightarrow \infty$.

\subsection{Quantum eigenvalue transformation}
In this section, we introduce the recently proposed quantum eigenvalue transformation (QEVT) framework for approximating the matrix function.
Specifically, a function of a matrix $A$ is acting on its eigenvalues:
$$f(A)=Pf(\Lambda)P^{-1},$$
where $A=P\Lambda P^{-1}$ is the eigen-decomposition of $A$, $f(x)$ is the targeted function. More specifically, the aim is to approximate matrix functions with polynomial transformations with (near) optimal performance in query complexity. That is, we want the number of calls to the matrix $A$ (e.g., accessed by the block encoding of $A$) to be as few as possible. As the times of query are closely related to the degree of the polynomial function, this is to say our objective is finding the (near) lowest degree of the polynomial function such that the error given by uniform approximation over the targeted region is desirable:
$$\min_{d}\sup_{x\in \mathbb{D}}|p_d(x)-f(x)|\leq \epsilon,$$
where $P_d(x)$ is a polynomial function of degree $d$, $S$ is the region for the eigenvalues of $A$, and $\epsilon$ is the accuracy we need to achieve. Typically, $\mathbb{D}:=\{z\in\mathbb{C}||z|\leq 1\}$ is the complex unit disk that encompasses all eigenvalues of the block encoded (normalized) matrix $\alpha^{-1}A$.

Using knowledge from approximation theory, the Chebyshev and Faber polynomials are introduced for the approximation of real and complex functions, respectively. The goal of the QEVT algorithm is to realize polynomial functions of matrix $A$ for approximating the target matrix function:
\begin{equation}\label{eq:target poly}
    p(A)=\sum_{j=0}^{d-1}\beta_jp_j(A),
\end{equation}
where $p_j(x)$ is some polynomial function of $x$ of degree $j$, and $\beta_j$ is coefficient. 
Then, we can apply the matrix function to an initial state $\ket{\psi_I}$ and prepare 
\begin{equation}\label{eq:target state}
    \frac{p(A)\ket{\psi_I}}{\|p(A)\ket{\psi_I}\|_2},
\end{equation}
where $\|p(A)\ket{\psi_I}\|_2$ is the normalization factor.
Yet, the challenges here are twofold: (i) How to realize the interested (Chebyshev or Faber) polynomial function natively on quantum circuits, and (ii) How to combine polynomial functions of different degrees with corresponding coefficients. The QEVT algorithm solves the first problem by resorting to the generating function of the targeted polynomial functions. The rationale is that the generating function encodes the polynomial functions as coefficients of a power series. Thus, by implementing the generating function on a quantum circuit, we may directly load each degree of polynomial functions onto the quantum devices rather than separately for each one of them. To see this, we follow the notation in Ref.~\cite{low2024quantum} and write the generating function of $\{p_j(x)\}_{j=0}^\infty$ as 
$$g(y,x)=\sum_{j=0}^\infty p_j(x) y^j,$$
which is a power series of $y$ with coefficient given by $p_j(x)$. While the sum is up to infinity, we can choose the matrix form of $y$ as the lower shift matrix:
\begin{equation}
    L=\sum_{k=0}^{d-2} \ket{k+1}\bra{k}
\end{equation}
so that $L^d=0$, resulting in the summation stopped at $d-1$. As such, we have the matrix representation of the generating function as
\begin{equation}
    \sum_{j=0}^{d-1}L_j\otimes p_j(A)=\sum_{j=0}^{\infty}L_j\otimes p_j(A)=g(L\otimes I, I\otimes A).
\end{equation}
Besides, it is readily seen that different degrees of the matrix polynomials $p_j(A)$ are encoded in different parts of the overall matrix identified by the power of $L$. 
The next step is to show that the generating functions of the two interested polynomial functions are all of simple forms: (i) For the real eigenvalues cases, the Chebyshev polynomial (denoted as $T_j(x)$) has the generating function as
$$\sum_{j=0}^\infty y^j T_j(x)= \frac{1}{2}\frac{1-y^2}{1-2yx+y^2},$$
and for the Faber polynomial (denoted as $F_j(z)$), we have
$$\sum_{j=0}^{\infty} {F}_j(z)y^j=\frac{{\Psi}^{\prime}\left(y^{-1}\right)}{y\left({\Psi}\left(y^{-1}\right)-z\right)}$$
such that $z \in \mathbb{D}$, ${\Psi}(\cdot)$ is the exterior Riemann map defined in Sec.~\ref{sec:cheby_faber} and $\Psi^\prime$ is its derivative. Then, the matrix Chebyshev generating function is given by
\begin{equation}\label{eq:cheby_gen}
    \sum_{j=0}^{n-1} L^j \otimes \widetilde{T}_j\left(\frac{A}{\alpha_A}\right)=\sum_{j=0}^{\infty} L^j \otimes \widetilde{T}_j\left(\frac{A}{\alpha_A}\right)=\frac{I \otimes I-L^2 \otimes I}{2\left(I \otimes I+L^2 \otimes I-2 L \otimes \frac{A}{\alpha_A}\right)}=:\frac{B_r}{A_r},
\end{equation}
where we have denoted the nominator and denominator as $B_r$ and $A_r$, respectively.
The matrix Faber generating function is given by
\begin{equation}\label{eq:faber_gen}
    \sum_{j=0}^{n-1} L^j \otimes {F}_j\left(\frac{A}{\alpha_A}\right)=\sum_{j=0}^{\infty} L^j \otimes {F}_j\left(\frac{A}{\alpha_A}\right)=\frac{\Psi^{\prime}\left(L^{-1}\right) \otimes I}{L \Psi\left(L^{-1}\right) \otimes I-L \otimes \frac{A}{\alpha_A}}=:\frac{B_c}{A_c},
\end{equation}
where we have denoted the nominator and denominator as $B_c$ and $A_c$, respectively.
Hence, we see that the matrix Chebyshev and Faber generating functions are both of rational polynomial function forms. The nominator of the function involves only functions of $L$, which can be readily block encoded as presented in~\cite[Lemma 14]{low2024quantum}. The denominators are implemented by acting $f(x)=x^{-1}$ with the variable being either $B_r$ or $B_c$ in each case, to some initial state. This problem is solved by the quantum linear system problem (QLSP) solver, where algorithms with optimal performance (query complexity) are known~\cite{costa2022optimal}.

The part of loading the coefficients of the polynomial function is merged into an input state preparation step. Straightforwardly, for a suitable choice of input state, we may act the matrix Chebyshev or Faber generating function onto this input state to produce the final target state as shown in Eq.~\eqref{eq:target state}. Yet, the key observation of the QEVT algorithm is that we can alternatively directly prepare the input state as the chosen state after the action of the nominator part of Eq.~\eqref{eq:cheby_gen} or \eqref{eq:faber_gen}. As the nominators ($B_r$ or $B_c$) in each case are of simple form, the analytic formula of the input states can be derived without triggering the explicit multiplication action of the nominators. Indeed, the input state in the real cases is given by
\begin{equation}\label{eq:final-initial-real}
    \ket{\psi_I'}:=\ket{0}\alpha_\beta'^{-1}\sum_{k=0}^{d-1}(\beta_k-\beta_{k+2})\ket{d-1-k}\ket{\psi_I},
\end{equation}
where $\alpha_\beta'=\sqrt{\sum_{k=0}^{d-1}(\beta_k-\beta_{k+2})^2}$, and $\beta_k$ are coefficients of the polynomial function as shown in Eq.~\eqref{eq:target poly}. For complex cases, this is given by
\begin{equation}\label{eq:final-initial-complex}
    \ket{\psi_I''}:=\ket{0}\alpha_\beta^{-1}\Psi^{\prime}\left(L_n^{-1}\right) \sum_{k=0}^{d-1} \beta_k|d-1-k\rangle,
\end{equation}
where 
\begin{equation}\label{eq:alpha beta}
    \alpha_\beta=\left\| \Psi^{\prime}\left(L_n^{-1}\right) \sum_{k=0}^{d=1} \beta_k|d-1-k\rangle \right\|.
\end{equation}
A Fourier coefficients generation can prepare this input state as discussed in \cite[Sec.~6]{low2024quantum}. 

The final key step is to generate the Chebyshev or Faber history state. The history state is the core of the algorithm design of the QEVT framework. This state consists of the desirable part given by Eq.~\eqref{eq:target state}, while superposed by other unwanted components, which are linear combinations of the polynomial function with mismatched coefficients. To boost the success probability in postselecting the favorable measurement outcome, the \emph{runaway padding trick}~\cite{berry2017quantum} is introduced.
This results in the desirable part in the superposed history state indicated by the first register. For the QEVT task, we observe that a single ancillary qubit in state $\ket{0}$ is employed, as shown in both Eq.~\eqref{eq:final-initial-real} and \eqref{eq:final-initial-complex}. For now, we provide the results for the complex cases to facilitate later discussion. The padded matrices are given by
\begin{equation}\label{eq:pad_def}
\begin{aligned}
\mathbf{Pad}(A_c)&= \ket{0}\bra{0} \otimes\left(L_d \Psi\left(L_d^{-1}\right) \otimes I-L_d \otimes \frac{A}{\alpha_A}\right) +|1\rangle\langle 0| \otimes|0\rangle\langle d-1| \otimes(-I)+|1\rangle\langle 1| \otimes\left(I_{\eta d}-L_{\eta d}\right) \otimes I ,\\
\mathbf{Pad}(B_c) & =|0\rangle\langle 0| \otimes \Psi^{\prime}\left(L_d^{-1}\right) \otimes I+|1\rangle\langle 1| \otimes I_{\eta d} \otimes I,
\end{aligned}
\end{equation}
where $\eta=1$. At this point, the history state is generated by the following steps:
\begin{enumerate}
    \item Construct the block encoding of the lower shift matrix $L$ by \cite[Lemma 14]{low2024quantum}.
    \item Using the block encoding of $L$ and $A$ for the construction of the block encoding of $\mathbf{Pad}(A_c)$.
    \item Prepare the input state as given by Eq.~\eqref{eq:final-initial-complex} and invoke the QLSP solver with the input $\mathbf{Pad}(A_c)$ for acting on the input state.
\end{enumerate}
As such, the cost of generating the Faber history state is as follows.
\begin{lemma}[Faber history state generation, Theorem 9, Ref.~\cite{low2024quantum}]
\label{lem:faber_history}
Let $A$ be a square matrix such that $A/\alpha_A$ is block encoded by $O_A$ with some normalization factor $\alpha_A\geq\norm{A}$. Suppose that eigenvalues of $A/\alpha_A$ are enclosed by a Faber region $\mathcal{E}$ with associated conformal maps ${\Phi}:\mathcal{E}^c\rightarrow\mathcal{D}^c$, ${\Psi}:\mathcal{D}^c\rightarrow\mathcal{E}^c$ and Faber polynomials ${F}_n(z)$.
Let $O_\psi\ket{0}=\ket{\psi}$ be the oracle preparing the initial state, and $O_{\beta}\ket{0}={\Psi}'(L_d^{-1})\sum_{k=0}^{d-1}\beta_k\ket{d-1-k}/\norm{{\Psi}'(L_d^{-1})\sum_{k=0}^{d-1}\beta_k\ket{d-1-k}}$ be the oracle preparing the shifting of coefficients $\beta$.
Then, the quantum state
	\begin{equation}\label{eq:faber_history}
		\frac{\ket{0}\sum_{l=0}^{n-1}\ket{l}
			\sum_{k=d-1-l}^{d-1}{\beta}_k{{F}}_{k+l-d+1}\left(\frac{A}{\alpha_A}\right)\ket{\psi}
			+\sum_{s=1}^{\eta}\ket{s}\sum_{l=0}^{d-1}\ket{l}
			\sum_{k=0}^{d-1}{\beta}_k{{F}}_{k}\left(\frac{A}{\alpha_A}\right)\ket{\psi}}
        {\sqrt{\sum_{l=0}^{d-1}\norm{\sum_{k=d-1-l}^{d-1}\beta_k{F}_{k+l-d+1}\left(\frac{A}{\alpha_A}\right)\ket{\psi}}^2
        +\eta n\norm{\sum_{k=0}^{d-1}\beta_k {F}_{k}\left(\frac{A}{\alpha_A}\right)\ket{\psi}}^2}}
	\end{equation}
can be prepared with accuracy $\epsilon$ and probability $1-\delta$ using
\begin{equation}
    \mathcal{O}\left(\alphaFP d(\eta+1)\log\left(\frac{1}{\epsilon}\right)\log\left(\frac{1}{\delta}\right)\right)
\end{equation}
queries to controlled-$O_A$, controlled-$O_\psi$, controlled $O_{\widetilde\beta}$, and their inverses, where
\begin{equation}
\label{eq:alphaFP_alphaPsiFaber}
    \alphaFP\geq\max_{j=1,\ldots,n}\norm{\frac{{F}_j'\left(\frac{A}{\alpha_A}\right)}{j}}
\end{equation}
is an upper bound on the derivative of Faber polynomials.
\end{lemma}
For constructing the QEVT algorithm, a choice of $\eta=1$ is enough, as stated before.
To reach the final desirable state as shown in Eq.~\eqref{eq:target state}, we remark that the history state already encodes this state as indicated by the first register with state $\ket{0}$. Hence, a fixed point amplitude amplification algorithm~\cite{yoder2014fixed} is applied to enhance the failure probability to be no more than $p_{\rm fail}$, such that a quadratic improvement is made compared to directly measuring the first qubit.

The design of the algorithm solving the QESE problem needs the block encoding version of the polynomial function. This is similar to the block encoding version of the QEVT algorithm in the real case as discussed in
\cite[Sec.~5.2]{low2024quantum}. We summarize the steps in the following.
\begin{enumerate}
    \item Construct the block encoding of $\mathbf{Pad}(A_c)/(2\alpha_\Psi+2)$, where $\alpha_{\Psi} \geq\left\|L_d \Psi\left(L_d^{-1}\right)\right\|$ for some $d\in \mathbb{Z}^+$. This block encoding can be realized by only $1$ query to the block encoding of $A$, as shown in \cite[Eq.~390-394]{low2024quantum}.
    \item Prepare the input state given by Eq.~\eqref{eq:final-initial-complex}.
    \item Invoke the block encoding version of the QLSP solver given by Lemma \ref{lem:inv_block} with input $\mathbf{Pad}(A_c)$. Unprepare $|1\rangle \frac{1}{\sqrt{d}} \sum_{k=0}^{d-1}|k\rangle$.
    \item Perform the block encoding amplification given by Lemma \ref{lem:amp_block}.
\end{enumerate}

\begin{lemma}[Quantum eigenvalue transformation with Faber polynomials, block encoded version]
\label{lem:qevt_c_block}
    Let $A$ be a square matrix such that $A/\alpha_A$ of block encoding $O_A$ with normalization factor $\alpha_A\geq\norm{A}$.
    Let $p(x)=\sum_{k=0}^{d-1}\beta_k{F}_{k}(x)$ be the Faber expansion of a degree-$(d-1)$ polynomial $p$.
    Then for any $\alpha_p\geq\norm{p\left(\frac{A}{\alpha_A}\right)}$, the operator
    \begin{equation}
        \frac{p\left(\frac{A}{\alpha_A}\right)}{2\alpha_p}
    \end{equation}
    can be block encoded with accuracy $\epsilon$ using
    \begin{equation}\label{eq:qevt_be_complexity_c}
        \begin{aligned}
        {\mathcal{O}\left(\frac{\alpha_\beta\sqrt{d}\alphaFP}{\alpha_p}d\alphaFP\log\left(\frac{\alpha_\beta\sqrt{d}\alphaFP}{\alpha_p\epsilon}\right)\log\left(\frac{1}{\epsilon}\right)\right)}
        \end{aligned}
    \end{equation}
    queries to controlled-$O_A$ and its inverse, where
    $\alpha_{\mathbf{F}^{\prime}}$ is upper bounded by Eq.~\eqref{eq:alphaFP_alphaPsiFaber}, and $\alpha_\beta$ is defined by Eq.~\eqref{eq:alpha beta}.
\end{lemma}
\begin{proof}
Using the block encoded version of matrix inversion given by Lemma \ref{lem:inv_block} with the input $\mathbf{Pad}(A_c)/(2\alpha_\Psi+2)$, we obtain the block encoding of 
$$\frac{\mathbf{Pad}(A_c)^{-1}}{2\alpha_{\scriptscriptstyle \mathbf{Pad}(A_c)^{-1}}},$$
such that $\alpha_{\scriptscriptstyle \mathbf{Pad}(A_c)^{-1}}$ is an upper bound on $\left\|\mathbf{Pad}(A_c)^{-1}\right\|$ and satisfies $\alpha_{\scriptscriptstyle \mathbf{Pad(A_c)}^{-1}}=\mathcal{O}(d\alphaFP)$, which is given by \cite[Eq.~400]{low2024quantum}. This factor decides the cost of the QLSP step. Besides, we remark that $\alpha_\Psi=\mathcal{O}(1)$ in general, as shown by \cite[Lemma 33]{low2024quantum}.

Now, the preparation of $\ket{0}\alpha_\beta^{-1}\Psi^{\prime}\left(L_n^{-1}\right) \sum_{k=0}^{d-1} \beta_k|d-1-k\rangle$ and unpreparation of $\ket{1}\frac{1}{\sqrt{d}}\sum_{k=0}^{d-1}\ket{k}$ gives us the block encoding in that subspace as
\begin{equation*}
\begin{aligned}
    &\left(\bra{1}\frac{1}{\sqrt{d}}\sum_{k=0}^{n-1}\bra{k}\otimes I\right)
    \frac{\mathbf{Pad}(A_c)^{-1}}{2\alpha_{\scriptscriptstyle \mathbf{Pad(A_c)}^{-1}}}
    \left(\ket{0}\alpha_\beta^{-1}\Psi^{\prime}\left(L_d^{-1}\right) \sum_{k=0}^{d-1} \beta_k|d-1-k\rangle\right)\\
    &=\frac{1}{\sqrt{d}\alpha_{\scriptscriptstyle \mathbf{Pad}(A_c)^{-1}}\alpha_{\beta}}
    \left(\sum_{k=0}^{d-1}\bra{k}\otimes I\right)
    \left(\sum_{l=0}^{d-1}\ket{l}
    \sum_{k=0}^{d-1}{\beta}_k F_{k}\left(\frac{A}{\alpha_A}\right)\otimes I\right)
    =\frac{p\left(\frac{A}{\alpha_A}\right)}{\alpha_{p,\text{pre}}}
\end{aligned}
\end{equation*}
with $\alpha_{p,\text{pre}}=\frac{\alpha_{\scriptscriptstyle \mathbf{Pad(A)}^{-1}}\alpha_{\beta}}{\sqrt{d}}=\mathcal{O}\left(\alphaFP \sqrt{d}\alpha_{\beta}\right)$. Here, in the second line, we have used the definition of $\mathbf{Pad}(B_c)$ in Eq.~\eqref{eq:pad_def}. Besides, we notice that up to a normalization factor, the state
$$\mathbf{Pad}(A_c)^{-1}\left(\ket{0}\alpha_\beta^{-1}\Psi^{\prime}\left(L_d^{-1}\right) \sum_{k=0}^{d-1} \beta_k|d-1-k\rangle\right)$$ prepares the Faber history state given by Eq.~\eqref{eq:faber_history}.
Thus, post-select on the first register to be $\ket{1}$ gives the Faber expansion. The accuracy of the QLSP protocol is of order $\left(\frac{\alpha_\beta\sqrt{d}\alphaFP}{\alpha_p\epsilon}\right)^{-1}$.
Finally, invoking the block encoding amplification to the accuracy of order $\epsilon$ gives us the final results.
\end{proof}

The real case is highly parallel with the above analysis, for which we only provide the results below and refer the readers to the original paper for a complete 
\begin{lemma}[Quantum eigenvalue transformation with Chebyshev polynomial, block encoded version, Theorem 5 of \cite{low2024quantum}]
\label{lem:qevt_r_block}
    Let $A$ be a square matrix with only real eigenvalues, such that $A/\alpha_A$ is block encoded by $U_A$ with some normalization factor $\alpha_A\geq\norm{A}$.
    Let $p(x)=\sum_{k=0}^{d-1}\widetilde\beta_k\widetilde{{T}}_{k}(x)=\sum_{k=0}^{d-1}\beta_k{{T}}_{k}(x)$ be the Chebyshev expansion of a degree-$(d-1)$ polynomial $p$.
    Then for any $\alpha_p\geq\norm{p\left(\frac{A}{\alpha_A}\right)}$, the operator
    \begin{equation}
        \frac{p\left(\frac{A}{\alpha_A}\right)}{2\alpha_p}
    \end{equation}
    can be block encoded with accuracy $\epsilon$ using
    \begin{equation}\label{eq:qevt_be_complexity_r}
        \begin{aligned}
        {\mathcal{O}\left(\frac{\norm{p(\cos)\sin}_{2,[-\pi,\pi]}\sqrt{d}\alphaU}{\alpha_p}d\alphaU\log\left(\frac{\norm{p(\cos)\sin}_{2,[-\pi,\pi]}\sqrt{d}\alphaU}{\alpha_p\epsilon}\right)\log\left(\frac{1}{\epsilon}\right)\right)}
        \end{aligned}
    \end{equation}
    queries to controlled-$U_A$ and its inverse, where $\alphaU$ satisfies 
    \begin{equation}\label{eq:alphaU}
        \alphaU \geq \max _{j=0,1, \ldots, d-1}\left\|{U}_j\left(\frac{A}{\alpha_A}\right)\right\|,
    \end{equation}
    where $U_j$ denotes the degree-$j$ Chebyshev polynomial of the second kind.
\end{lemma}

\section{Quantum eigenvalue estimation problem}
In this section, we introduce the quantum estimation of dominant eigenvalues (QESE) problem and develop algorithms for solving this problem. The algorithms are essentially based on the matrix pencil (MP) methods~\cite{sarkar1995using,hua1990matrix} from classical signal processing, as we will introduce in the next section. 

\subsection{Problem statement}
As we will show later, it is necessary for the signal as given by Eq.~\eqref{eq:signal} in the MP method to be sparse. That is $r=\mathcal{O}(1)$ when the eigenvalues are complex; otherwise, if the eigenvalues are real, we can allow $r=\mathcal{O}(\mathrm{poly}(n))$. Thus, in the quantum setting, we make the following assumption that an initial density matrix is known such that the density matrix has a sparse expansion in the (diagonal) eigenbasis of the targeted matrix $A$. 
\begin{assumption}[Assumptions on the target matrix and sparsity of the initial density matrix]\label{assm:sparsity}
Let $A$ be the given matrix with eigen-decomposition $A=\sum_i \lambda_i \ket{\psi_i}\bra{\phi_i}$ that $\ket{\psi_i}$ and $\ket{\phi_i}$ are biorthogonal bases satisfying $\bra{\psi_i}\phi_j\rangle=\delta_{ij}$. 
We assume a $(\alpha_A, m, 0)$-block encoding $U_A$ of $A$ is given with $\alpha_A$. Besides, we assume that $A$ is block encoded in upper left corner of $U_A$ such that $\ket{G_1}=\ket{G_2}=\ket{0^m}$ in Eq.~\eqref{eq:be}.
Assume that an initial density matrix $\rho_I$ is known such that when expanded in the eigenbasis of $A$, the matrix is $r$-sparse for the diagonal part:
\begin{equation}
    \rho_I=\sum_{i=1}^r c_i \ket{\psi_i}\bra{\phi_i} + \sum_{i,j:i\neq j} c_{ij} \ket{\psi_i}\bra{\phi_j},
\end{equation}
where $c_i$ and $c_{ij}$ are the expanded coefficients. Besides $\sum_{i=1}^r c_i=1$ due to $\tr{\rho_I}=1$. Also, denote the coefficient of minimal absolute value to be
\begin{equation}\label{eq:c_min}
    c_{\min}=\min_{i\in[r]}|c_i|.
\end{equation}
\end{assumption}

Our assumption for the BE of $A$ is common in the literature of quantum singular value transformation~\cite{gilyen2019quantum}, such that we only specify the explicit form of $\ket{G_1}$ and $\ket{G_2}$ for the convenience of later discussion. Besides, the knowledge of the normalization factor $\alpha_A$ is known if we know how to construct the BE of $A$, such as through a linear combination of unitaries~\cite{childs2017quantum,low2019hamiltonian} or quantum walk~\cite{berry2009black}.

Intriguingly, an arbitrary $\rho$ can be expanded in the eigenbasis of $A$ because the eigenbasis is complete, such that $\sum_i \ket{\psi_i}\bra{\phi_i}=I$. By virtue of generalized eigenvectors as discussed in Sec.~\ref{sec:basic_matrix}, this property holds even in cases when $A$ is non-diagonalizable.
\begin{remark}
It is worth noting that the initial density matrix could be a pure state. In this case, we denote it as $\rho_I=\ket{\psi_I}\bra{\psi_I}$, and the assumption becomes
\begin{equation}
    \rho_I=\ket{\psi_I}\bra{\psi_I}=\sum_{i,j}\bra{\phi_j}\psi_I\rangle \bra{\psi_I}\psi_j\rangle \ket{\psi_i}\bra{\phi_j}=\sum_{i=1}^r c_i\ket{\psi_i}\bra{\phi_i}+\sum_{i,j:i\neq j}c_{ij} \ket{\psi_i}\bra{\phi_j},
\end{equation}
where $c_i:=\bra{\phi_i}\psi_I\rangle \bra{\psi_I}\psi_i\rangle$, and $c_{ij}:=\bra{\phi_i}\psi_I\rangle \bra{\psi_I}\psi_j\rangle$.

We demand the query access of $U_I$ such that $\ket{\psi_I}=U_I\ket{0}$ when the initial state is pure. Otherwise, we demand sample access to $\rho_I$, that is, the ability to prepare $\rho_I$ repeatedly.
\end{remark}

We then define the quantum estimation of sparse eigenvalues (QESE) problem.
\begin{definition}[The quantum estimation of sparse eigenvalues problem]
Given the Assumption \ref{assm:sparsity}, we find 
$$\tr{\rho_I A}=\sum_{i=1}^r c_i\lambda_i.$$
The goal is to output an estimation, denoted as $\widetilde{\lambda}_i$, of each $\lambda_i$ presented in the result of the trace such that
\begin{equation}
    |\widetilde{\lambda}_i-{\lambda}_i|<\epsilon.
\end{equation}
\end{definition}

In the remainder of the paper, we will interchangeably refer to $\rho_I$ as the initial density matrix or states. Typically, any quantum algorithm for solving the QESE problem will only be efficient if $c_{\min}=\Omega(1/\mathrm{poly}(n))$.
The demand for not-too-small components ($c_i$) in the initial state is natural in both classical and quantum settings. This is because when the amplitude $c_i$ vanishes, the cost of discriminating the signal from the initial state may scale $\rm poly(|c_i|^{-1})$. This is shown from the ground-state energy estimation problem, which shows that a lower-bounded estimation of the overlap with the ground state is needed.


\section{Main algorithm}

\subsection{Classical signal processing methods}
We first introduce the celebrated Prony's method along with some useful properties. We then apply these properties and introduce the MP methods~\cite{sarkar1995using,hua1990matrix}.

The goal of Prony's method is to estimate each component $c_j$ and $\omega_j$ (or equivalently $z_j$) in a given low-rank signal
\begin{equation}\label{eq:signal}
    q(t)=\sum_{j=1}^r c_j e^{i\omega_j t}=\sum_{j=1}^r c_jz_j^t,
\end{equation}
with $c_j\in\mathbb{C}\setminus\{0\}$ using as few as possible samples for $q(t)$ for different $t$. For $\omega\in\mathbb R$, the problem becomes super-resolution and (sparse) Fourier analysis problems. For $\omega\in\mathbb C$, the assumption that the signal is not increasing in amplitude is made, i.e.~$\rm Re(i\omega)\leq0$, which is the damping factor.

This method originates from Prony's method~\cite{hauer1990initial}, where the key gradient is the shift operator $S$. Consider the shift operator $\hat{S}$ that satisfies
$$\hat{S}z_j^t=z_j^{(t+1)}$$
for $j$. The shift operator is characterized by the eigenfunction 
$$(\hat{S}-z_jI)z_j^t=0,$$
where $I$ is the identity operator. Then, by linearity we find that $\hat{S}q(t)=q(t+1)$. We now define the Prony polynomial, which is the polynomial with $e^{i\omega_k }$ as roots: $P(z):=\prod_{k=1}^{r}(z-z_k)=\sum_{k=0}^{r-1}\alpha_k z^k+z^r$. It is then found that for arbitrary $t$, we have
\begin{equation}
\begin{aligned}
    P(S)q(t)&=\sum_{k=0}^r\alpha_k \hat{S}^k \sum_{j=1}^r c_j z_j^t\\
    &=\sum_{j=1}^r c_j z_j^t\sum_{k=0}^r \alpha_k z_j^k\\
    &=\sum_{j=1}^r c_j z_j^t P(z_j)=0.
\end{aligned}
\end{equation}
By knowledge of the Prony polynomial, we know that $\alpha_r=1$, resulting in the following matrix expression
\begin{equation}\label{eq:linear_system}
    H_0(\alpha_k)_{k=0}^{r-1}=-(q(r+j))_{j=0}^{r-1},
\end{equation}
$H_0\in \mathbb{C}^r\times\mathbb{C}^{r}$ is the \emph{square Hankel matrix} defined as $H_x:=(q(x+j+k))_{j,k=0}^{r-1}$, where the $x=0$ here is reflected as a shift in the first element in the matrix. The Hankel matrix has rank $r$, involves equispaced signals $q(j),~j\in[2r]-1$, and has the decomposition, which we refer to as the Vandermonde decomposition:
\begin{equation}\label{eq:h0}
    H_0=V(\mathbf{z}) C V(\mathbf{z})^{\rm T},
\end{equation}
$V(\mathbf{z}):=(z_j^{k-1})_{j,k=1}^r$ is the \emph{Vandermonde matrix} generating from $\textbf{z}:=(z_j)_{j=1}^r$, and $C:=\rm{diag}(\mathbf{c})$ is a diagonal matrix generating from $\mathbf{c}:=(c_j)_{j=1}^r$, which are coefficients in Eq.~\eqref{eq:signal}.

Prony's method then works by solving the linear system given by Eq.~\eqref{eq:linear_system} for the vector $\bm{\alpha}$. Subsequently, each $z_j$ can be obtained by finding the roots for $P(z)=0$. Finally, the coefficients are solved by the Vandermonde system $$V(\mathbf{z})(c_j)_{j=1}^r=(q(k))_{k=0}^{r-1}.$$
In practice, an upper bound $R$ on the rank is usually known, and one can learn the signal's sparsity by constructing the Hankel matrix and probing the rank. However, because of the root-finding step, Prony's method may suffer from instability due to noise.

To address this problem, we observe that the key insight in Prony's method is to leverage the shift invariance of the sum of exponential structure. That is, equispaced samples of the signal given by Eq.~\eqref{eq:signal} reside in the same subspace regardless of the starting point of the samples. The problem then becomes extracting information from such a shift-invariant subspace. As such, we examine the Vandermonde decomposition of $H_1$ as
\begin{equation}\label{eq:h1}
    H_1=V(\mathbf{z}) CZ V(\mathbf{z})^{\rm T},
\end{equation}
where $Z:=\rm{diag}(\bm{z})$ generating from $\bm{z}$. Then, we define the \emph{matrix pencil} as the family of matrices $H(1)-\lambda H(0)$ parameterized by $\lambda$. We claim that \emph{the poles of the matrix pencil give $\lambda=z_j$} such that the signal processing problem is addressed. To justify this, we consider the generalized eigenvalue problem (GEVP) as follows. 
\begin{definition}[Generalized eigenvalue problem]\label{def:gevp}
Let $H_0, H_1\in \mathbb{C}^{r\times r}$ be Hankel matrices that are provided by Eq.~\eqref{eq:h0} and Eq.~\eqref{eq:h1} with each element a signal given by Eq.~\eqref{eq:signal}. The generalized eigenvalue problem is defined by solving the following equation for each generalized eigenvalue $\lambda$:
\begin{equation}\label{eq:generalized_ev}
    H_1\ket{\nu}=\lambda H_0\ket{\nu}.
\end{equation}
\end{definition}
For the above equation, by plugging in the Vandermonde decomposition, we obtain
\begin{align*}
    V(\mathbf{z}) CZ V(\mathbf{z})^{\rm T}\ket{\nu}&=\lambda V(\mathbf{z}) C V(\mathbf{z})^{\rm T}\ket{\nu}\\
    Z\left(V(\mathbf{z})^{\rm T}\ket{\nu}\right)&=\lambda  \left(V(\mathbf{z})^{\rm T}\ket{\nu}\right).
\end{align*}
Here, the last line yields the eigenfunction that is satisfiable by eigenvalues $\lambda=z_j$. This is achieved by multiplying $(V(\mathbf{z}) C)^{-1}$ on both sides, as $V$ is of full rank and $C$ contains non-singular diagonal values. As such, by solving the generalized eigenvalue problem given by Eq.~\eqref{eq:generalized_ev}, we solve the signal processing problem. This completes the description of the MP method we used in this work.

\begin{remark}
In a practical sense, we may not directly know the sparsity $r$ of the signal. Yet, we may have the knowledge about the upper bound on the sparsity. Then, it is possible to construct the Hankel matrix and compute its rank and acquire the sparsity of the signal. Thus, for simplicity of the discussion, we assume knowledge of the sparsity. But an upper bound to it is equally good as long as the upper bound is on the same order as the true sparsity.
\end{remark}

Moreover, we compare the MP method here to that given in \cite[Algorithm 5.1]{dutkiewicz2022heisenberg}. The algorithm defines the shift matrix
$$S:=V(\mathbf{z})ZV^{-1}(\mathbf{z}),$$
which obeys $SH_0=H_1$. As such, one attains the full set of $\{z_j\}_{j=1}^r$ from eigenvalues of $S$, which is approximated by $G$ by solving
$$\min\|GH_0-H_1\|.$$
The two MP methods capture different aspects of the shift-invariance of the signal, while our method is conceptually simpler. And the authors only consider the tasks of phase estimation. Besides, as we will see in sections, our method yields noise-tolerant results for both sampling and noise in signals.

\subsection{Quantum algorithms combining the QEVT protocol}\label{sec:q_alg}
Our scheme for solving the QESE problem is to first generate the signal in the form of Eq.~\eqref{eq:signal} by a quantum computer and then solve for each $z_j$ in the signal to give us the information about the eigenvalues we aim for.

A simple idea for constructing such a signal may be realized by estimating the trace of $\rho_I$ with the power function of $A$:
\begin{equation}\label{eq:q_signal}
    g(t)=\tr{\rho_I A^t}=\alpha_A^t\tr{\rho_I(A/\alpha_A)^t}=\tr{\sum_{i} c_i \bra{\phi_i}A^t\ket{\psi_i}}=\sum_{i=1}^r c_i \lambda_i^t,
\end{equation}
where we have used the cyclic property of trace and biorthogonality of the eigenbasis of $A$. The above equation also indicates ways to approximate the signal: (i) First, construct the matrix function as $(A/\alpha)^t$ using Lemma \ref{lemma:block_prod}, (ii) Then, estimate the trace between the matrix function and the initial density matrix, and (iii) Multiply back the normalization factor $\alpha_A^t$ to get an final estimation to Eq.~\eqref{eq:q_signal}. Besides, when comparing Eq.~\eqref{eq:q_signal} to Eq.~\eqref{eq:signal}, we find that $z_i=\lambda_i$. Thus, the output of the GEVP problem is directly each $\lambda_i$.

In particular, the trace between $(A/\alpha)^t$ and the initial density matrix can be estimated using the Hadamard circuit as shown in Fig.~\ref{fig:main} (a) with the controlled unitary being a $(\alpha_A^t,mt,0)$-block encoding of the $A^t$, denoted as $U_{A^t}$. That is when $W=I$ and $S^\dagger$, the measurement outcome $\{X_j\}_j$ and $\{Y_j\}_j,~X_j, Y_j\in\{-1,1\}$ are unbiased estimators to the real and imaginary parts of the signal given by Eq.~\eqref{eq:q_signal} such that
\begin{equation}\label{eq:power_unbiased}
\begin{aligned}
    \mathbb{E}[X_j]&=\frac{1}{2}\left(\tr{\rho_I(\bra{0^m}U_{A^t}\ket{0^m}+\bra{0^m}U_{A^t}^\dagger\ket{0^m})}\right)=\tr{\rho_I\frac{A^t+(A^t)^\dagger}{2\alpha_A^t}}=\frac{\tr{\rho_I A^t_h}}{\alpha_A^t}=\mathrm{Re}(g(t)/\alpha_A^t)\\
    \mathbb{E}[Y_j]&=\frac{1}{2i}\left(\tr{\rho_I(\bra{0^m}U_{A^t}\ket{0^m}-\bra{0^m}U_{A^t}^\dagger\ket{0^m})}\right)=\tr{\rho_I\frac{A^t-(A^t)^\dagger}{2\alpha_A^ti}}=\frac{\tr{\rho_I A^t_a}}{\alpha_A^t}=\mathrm{Im}(g(t)/\alpha_A^t),
\end{aligned}
\end{equation}
where $\mathrm{Re}(x)$ and $\mathrm{Im}(x)$ denotes the real and imaginary part of $x$, respectively; $A_h^t$ and $iA_a^t$ denotes the hermitian and anit-hermitian parts of $A^t$ as given by Eq.~\eqref{eq:ha_decomp}. Note that $U^\dagger_{A_t}$ block encodes $(A^\dagger)^t$. 
Therefore, one may estimate the signal $g(t)$ by multiplying $\alpha^t_A$ by the empirical mean of the samples. Besides, Eq.~\eqref{eq:power_unbiased} is justified by
$$\tr{\rho_I A^t_h}+i \tr{\rho_I A^t_a}=\tr{\rho_IA^t}=g(t),$$
along with the fact that the spectra of hermitian and anti-hermitian matrices are real and imaginary, respectively. Thus, it holds that 
$$\tr{\rho_I A^t_h}\equiv\mathrm{Re}(g(t)),~\tr{\rho_I A^t_a}\equiv\mathrm{Im}(g(t)),$$
because of $\rho_I$ is hermitian.

For more complicated cases, the signal could be given by taking the trace of some matrix function $f(A)$ and the initial state:
\begin{equation}\label{eq:q_signal_f}
    g^\prime(t)=\tr{\rho_If_t(A)}\tr{\sum_{i} c_i \bra{\phi_i}f_t(A)\ket{\psi_i}}=\sum_{i=1}^r c_i f_t(\lambda_i).
\end{equation}
The matrix function can then be approximated by a polynomial function using the QEVT algorithm. Specifically, the first step it to construct a $(2\alpha_p, m', \varepsilon)$-block encoding unitary $U_{p(A)}$ of $p(A/\alpha_A)$ as given by Lemma  \ref{lem:qevt_c_block} and \ref{lem:qevt_r_block} of $f(A)$ for some $m'$ and $\varepsilon$. Although the matrix is normalized in the block encoding, we may still approximate $f(A)$ for $A$ with a much larger spectral radius than $1$ by assuming the map is dissipative such that $f(A)$ is a trace non-increasing operator.

For signal given by Eq.~\eqref{eq:q_signal_f}, similar to the power function case, we have for the Hadamard test circuit by controlling $U_{p(A)}$ with $W=I$ and $S^\dagger$, the measurement outcome $\{X_j\}_j$ and $\{Y_j\}_j,~X_j, Y_j\in\{-1,1\}$ are unbiased estimators to the real and imaginary parts of the signal such that
\begin{equation}\label{eq:unbaised_p}
\begin{aligned}
    \mathbb{E}[X_j]&=\frac{1}{2}\left(\tr{\rho_I(\bra{s}U_{p(A)}\ket{\psi_I''}+\bra{s}U_{p(A)}^\dagger\ket{\psi_I''})}\right)=\tr{\rho_I\frac{p(A/\alpha_A)+p^\dagger(A/\alpha_A)}{4\alpha_p}}=\frac{\tr{\rho_I p_h(A/\alpha_A)}}{2\alpha_p}\approx\mathrm{Re}(g'(t)/(2\alpha_p)),\\
    \mathbb{E}[Y_j]&=\frac{1}{2}\left(\tr{\rho_I(\bra{s}U_{p(A)}\ket{\psi_I''}-\bra{s}U_{p(A)}^\dagger\ket{\psi_I''})}\right)=\tr{\rho_I\frac{p(A/\alpha_A)-p^\dagger(A/\alpha_A)}{4\alpha_p i}}=\frac{\tr{\rho_I p_a(A/\alpha_A)}}{2\alpha_p}\approx\mathrm{Im}(g'(t)/(2\alpha_p)),
\end{aligned}
\end{equation}
where $p_h(A/\alpha_A):=\frac{p(A/\alpha_A)+p^\dagger(A/\alpha_A)}{2}$, $p_a(A/\alpha_A):=\frac{p(A/\alpha_A)-p^\dagger(A/\alpha_A)}{2i}$, $\ket{\psi_I''}$ is defined by Eq.~\eqref{eq:final-initial-complex}, and $\ket{s}:=\frac{1}{\sqrt{d}}\ket{1}\sum_{k=0}^{d-1}\ket{k}$. The final approximate equality results from $p(A/\alpha_A)$ only approximating $f_t(A)$. Again, up to a normalization factor, we attain the real and imaginary parts of the signal.

The above discussion necessitates several modifications of the Hadamard test circuit. First, because of the distinct subspace for the polynomial function encoded in $U_{p(A)}$, we need to prepare a special input state before applying the controlled block encoding unitary. This is illustrated in Fig.~\ref{fig:main} (a) of the unitary that depends on the control qubit of the Hadamard circuit to be either $\ket{0}$ or $\ket{1}$, prepares $\ket{s}$ or $\ket{\psi_I''}$ for the ancillary register of the block encoding unitary with input state $\rho_I$. In the next section, we show that the Hadamard test protocol provides an unbiased estimator to the real and imaginary parts of the signal as given by Eq.~\eqref{eq:unbaised_p} using the language of purification.


In more general cases, the normalization factor $\alpha$ is not one, and we may also be interested in matrix functions of $A$ other than the power function, such as the exponential function that typically appears as the solution to ordinary differential equations and open quantum systems. The matrix function is then realized by the QEVT algorithm. We discuss these problems in detail in Sec.~\ref{sec:app}.



\subsection{Quadratic speedup with purified quantum query access}
When the purification of the initial density matrix is known, we could invoke the QAE protocol instead of the Hadamard test algorithm. To this end, we first define the purified quantum query access of a density matrix.
\begin{definition}[Purified quantum query access]\label{def:purified}
Let $\rho \in \mathbb{C}^{N \times N}$ with $N=2^n$ be a density matrix of eigen-decomposition $\rho=\sum_{i=1}^N p_i\ket{\varphi_i}\bra{\varphi_i}$. It has purified quantum query access if we have access to the unitary oracle $U_\rho$ such that
\begin{equation}\label{eq:purified state}
    U_\rho|0\rangle_E|0\rangle_I=\left|\rho\right\rangle_{E I}=\sum_{i=1}^N \sqrt{p_i}\left|\phi_i\right\rangle_E\left|\varphi_i\right\rangle_I, \quad \text { where }\left\langle\phi_i | \phi_j\right\rangle=\left\langle\varphi_i | \varphi_j\right\rangle=\delta_{i j}.
\end{equation}
The density matrix then acts as the reduced density matrix of the purified state $\operatorname{Tr}_E\left(\left|\rho\right\rangle\left\langle\rho\right|\right)=\rho$.
\end{definition}
When given the purified quantum query access of the initial density matrix, denoted as $\ket{\rho_I}$, we show that it is possible to apply the AE protocol to enhance the Hadamard test quadratically with the scaling on the accuracy. 

First, consider the block encoding unitary \(U_{p(A)}\) acting on an ancilla register \(b\) and the system register such that $
(\langle 0|_b \otimes I) U_{p(A)} (|0\rangle_b \otimes I) = \frac{p(A/\alpha_A)}{2\alpha_p}$. The implementation of the block encoding through the QEVT framework is given by Lemma \ref{lem:qevt_c_block}. Hence, the objective can be estimate through the block encoded polynomial function up to some factors: $\text{Tr}(\rho_I p(A/\alpha_A)) = 2\alpha_p  \text{Tr}\left( \rho_I \frac{p(A/\alpha_A)}{2\alpha_p} \right)$. 
We now assume the factor $2\alpha_p p(\alpha_A)$ to be known and focus on estimating the trace between $\rho_I$ and the block encoded polynomial function.

Our circuit construction is similar to that of the Hadamard test circuit, adapted for mixed states and block encodings. The circuit comprises a control qubit $c$, an ancilla register $b$, and registers $E$ and $I$. Starting from the initial state $|0\rangle_c |0\rangle_b |\rho_I\rangle_{E I}$, we proceed as follows to estimate the real part of the trace.
\begin{equation}\label{eq:circ_hadamard_test}
\begin{alignedat}{2}
\ket{0}_c\ket{0}_b\ket{\rho_I}_{EI} &\xrightarrow{H_c} &\quad& \frac{|0\rangle_c + |1\rangle_c}{\sqrt{2}} \otimes |0\rangle_b |\rho_I\rangle_{EI} \\
&\xrightarrow{U_s} &\quad& \frac{1}{\sqrt{2}}\left[\ket{0}_c\ket{s}_b\ket{\rho_I}_{EI}+\ket{1}_c\ket{\psi_I''}_b\ket{\rho_I}_{EI}\right] \\
&\xrightarrow{\rm{c}-U_{p(A)}} &\quad& \frac{1}{\sqrt{2}} \left[ |0\rangle_c |s\rangle_b |\rho_I\rangle_{E I} + |1\rangle_c \left( |s\rangle_b \frac{p(A/\alpha_A)_I}{2\alpha_p} |\rho_I\rangle_{E I} + |\Phi^\perp\rangle_{b E I} \right) \right] \\
&\xrightarrow{H_c} &\quad& \quad~|0\rangle_c \otimes \frac{1}{2} \left[ |s\rangle_b |\rho_I\rangle_{E I} + |s\rangle_b \frac{p(A/\alpha_A)_I}{2\alpha_p} |\rho_I\rangle_{E I} + |\Phi^\perp\rangle_{b E I} \right] \\
&&\quad& + |1\rangle_c \otimes \frac{1}{2} \left[ |s\rangle_b |\rho_I\rangle_{E I} - |s\rangle_b \frac{p(A/\alpha_A)_I}{2\alpha_p} |\rho_I\rangle_{E I} - |\Phi^\perp\rangle_{b E I} \right].
\end{alignedat}
\end{equation}
Here, the $H_c$ means apply the Hadamard gate on the register $c$. $U_s$ is defined as 
\begin{equation}
  U_s = \ket{0}\bra{0}_c \otimes V_s + \ket{1}\bra{1}_c \otimes V_I,
  \qquad \text{where } V_s\ket{0}_b = \ket{s}_b \text{ and } V_I\ket{0}_b = \ket{\psi_I''}_b.
  \label{eq:Us_controlled_prep} 
\end{equation}
$|\Phi^\perp\rangle_{b E I}$ is the state that orthogonal to $\ket{0}_b$ and $\mathrm{c}-U_{p(A)}$ denotes the controlled $U_{p(A)}$ with control register $c$.

Next, measure register $c$ in the state $\ket{0}_c$ gives probability:
$$
P(|0\rangle_c) = \frac{1}{4} \left\| |\rho_I\rangle_{E I} + \frac{p(A/\alpha_A)_I}{2\alpha_p} |\rho_I\rangle_{E I} \right\|_2^2 + \frac{1}{4} \langle \Phi^\perp | \Phi^\perp \rangle,
$$
where $\|\cdot\|_2$ is the $\rm{L}_2$ norm. The two terms in the above summation are expanded as
\begin{equation*}
\begin{aligned}
    \left\| |\rho_I\rangle_{EI} + \frac{p(A/\alpha_A)_I}{2\alpha_p} |\rho_I\rangle_{EI} \right\|_2^2 &= 1 + \frac{1}{4\alpha_p^2} \langle \rho_I | p(A/\alpha_A)^\dagger p(A/\alpha_A) |\rho_I \rangle + \frac{1}{2\alpha_p} \left( \langle \rho_I | p(A/\alpha_A) |\rho_I \rangle + \langle \rho_I | p(A/\alpha_A)^\dagger |\rho_I \rangle \right)\\
    \langle \Phi^\perp | \Phi^\perp \rangle &= 1 - \frac{1}{4\alpha_p^2} \langle \rho_I | p(A/\alpha_A)^\dagger p(A/\alpha_A) |\rho_I \rangle,
\end{aligned}
\end{equation*}
where the second line is resulting from $\left\||s\rangle_b \frac{p(A/\alpha_A)_I}{2\alpha_p} |\rho_I\rangle_{E I} + |\Phi^\perp\rangle_{b E I}\right\|_2=1$ due to the state is normalized. This yields
$$
P(|0\rangle_c) = \frac{1}{2} + \frac{\text{Re} \left( \text{Tr}\left( \rho_I p(A/\alpha_A) \right) \right)}{4\alpha_p}.
$$
Hence, it finally gives us the estimator for the real part of the trace as
\begin{equation}\label{eq:ae_real}
    \text{Re} \left( \text{Tr}\left( \rho_I p\left(A/\alpha_A\right) \right) \right) = 2\alpha_p (2 P(|0\rangle_c) - 1).
\end{equation}

Similarly, for the imaginary part, we insert a phase gate $S^\dagger$ before the second Hadamard gate, giving the state:
$$
\frac{1}{\sqrt{2}} \left[ |0\rangle_c |s\rangle_b |\rho_I\rangle_{E I} - i |1\rangle_c \left( |s\rangle_b \frac{p(A/\alpha_A)_I}{2\alpha_p} |\rho_I\rangle_{E I} + |\Phi^\perp\rangle_{b E I} \right) \right].
$$
After applying the Hadamard gate, the probability to attain the measurement outcome $\ket{0}_c$ is
$$
P'(|0\rangle_c) = \frac{1}{4} \left\| |\rho_I\rangle_{E I} - i \frac{p(A/\alpha_A)_I}{2\alpha_p} |\rho_I\rangle_{E I} \right\|_2^2 + \frac{1}{4} \langle \Phi^\perp | \Phi^\perp \rangle = \frac{1}{2} + \frac{\text{Im} \left( \text{Tr}\left( \rho_I p(A/\alpha_A) \right) \right)}{4\alpha_p},
$$
The estimator for the imaginary part of the trace has the same form as the real part, which is
\begin{equation}\label{eq:ae_img}
    \text{Im} \left( \text{Tr}\left( \rho_I p\left(A/\alpha_A\right) \right) \right) = 2\alpha_p (2 P'(|0\rangle_c) - 1).
\end{equation}

Therefore, it is now clear that we can estimate the amplitude for $P(|0\rangle_c)$ and $P'(|0\rangle_c)$ to precision $\frac{\varepsilon}{4  \alpha_p}$ for some $\varepsilon$ using QAE as given by Lemma \ref{lemma:ae} with $\widetilde{O}\left( \frac{\alpha_p }{\varepsilon} \right)$ queries to $U_s$ and $\rm{c}-U_{p(A)}$ and their inverse. Besides, we note that if the matrix function $p(A/\alpha_A)$ is trace non-increasing, we can always set $\alpha_p=\mathcal{O}(1)$.

Eq.~\eqref{eq:unbaised_p}.
Besides, using the language of purification, we can readily formulate the estimator for the Hadamard test circuit.
\begin{remark}
The estimator for the real and imaginary parts of $\tr{\rho_I p(A/\alpha_A)}$ are given by
\begin{equation}\label{eq:hadamard_estimator1}
\begin{aligned}
    \mathrm{Re} \left( \mathrm{Tr}\left( \rho_I p\left(A/\alpha_A\right) \right) \right)&=2\alpha_p(P(|0\rangle_c)-P(|1\rangle_c)),\\
    \mathrm{Im} \left( \mathrm{Tr}\left( \rho_I p\left(A/\alpha_A\right) \right) \right)&=2\alpha_p(P'(|0\rangle_c)-P'(|1\rangle_c)),
\end{aligned}
\end{equation}
where $P(|x\rangle_c)$ and $P'(|x\rangle_c)$ are probabilities of register $c$ in Fig.~\ref{fig:main} (b) measures $x$ for $x\in\{0,1\}$  with $W=I$ and $S^\dagger$, respectively. In particular, for the signal given by Eq.~\eqref{eq:q_signal}, we have
\begin{equation}\label{eq:hadamard_estimator2}
\begin{aligned}
    \mathrm{Re} \left( \mathrm{Tr}\left( \rho_I A^t \right) \right)&=\alpha_A^t(P(|0\rangle_c)-P(|1\rangle_c)),\\
    \mathrm{Im} \left( \mathrm{Tr}\left( \rho_I A^t \right) \right)&=\alpha_A^t(P'(|0\rangle_c)-P'(|1\rangle_c)).
\end{aligned}
\end{equation}
\end{remark}
To explain, we find that tracing out the register $I$ in Eq.~\eqref{eq:circ_hadamard_test} will provide the correct state of the Hadamard test circuit. This reflects that the probabilities from measuring the register $c$ with or without purification are the same. This then gives us a similar estimator for the signals for the Hadamard test and the QAE algorithm. Besides, for the signal in Eq.~\eqref{eq:q_signal}, we construct the BE of $(A/\alpha_A)^t$ using Lemma \ref{lemma:block_prod}. This then gives the measurement probabilities that satisfy $P(|0\rangle_c)-P(|1\rangle_c))=\text{Re} \left( \text{Tr}\left( \rho_I (A/\alpha_A)^t \right) \right)$ and similarly for the imaginary part. This concludes Eq.~\eqref{eq:hadamard_estimator2}.

\section{Complexity analysis}

\subsection{Error propagation of the Vandermonde decomposition}
We provide the noise-tolerant result of the MP method and analyze the sample complexity of the algorithm. It is worth noting that the ideal GEVP given by Eq.~\eqref{eq:generalized_ev} can be reformulated into an eigenvalue problem by multiplying $H^{-1}_0$ on both sides of the equation:
$$H^{-1}_0H_1\ket{\nu}=\lambda\ket{\nu}.$$
By denoting $G:=H^{-1}_0H_1$, we find its eigen decomposition as $G=(V^{\rm T})^{-1}ZV^{\rm T}$, where we have denoted $V(\bm{z})$ as $V$ for notational convenience. 

For practical cases, we denote the noisy Hankel matrix as $\widetilde{H}_0:=H_0+E_0$ and $\widetilde{H}_1:=H_1+E_1$, where $E_0$ and $E_1$ are noise matrices. Without loss of generality, we first assume that the noise matrices consist of Gaussian elements (i.e., induced by sampling noises) and the noise rates are on the same level for both matrices such that their operator norms satisfy $\|E_0\|=\|E_1\|=\|E\|$. As such, the noisy GEVP problem can be similarly treated as an eigenvalue problem
\begin{equation}\label{eq:noise_general}
\begin{aligned}
    \widetilde{H}_1\ket{\widetilde\nu}&=\widetilde\lambda\widetilde{H}_0\ket{\widetilde\nu}\\
    \widetilde{G}\ket{\widetilde\nu}&=\widetilde\lambda\ket{\widetilde\nu},
\end{aligned}
\end{equation}
where $\widetilde{G}:=\widetilde{H}^{-1}_0\widetilde H_1$. Then, our goal becomes to deduce how noises impact the eigenvalues, which is determined by the change of $G$.
To this end, we resort to the first-order perturbation theory and denote $\tilde{G}=G+\Delta G$:
\begin{align*}
\widetilde{G} &= \left(H_0+E_0\right)^{-1}(H_1+E_1) \\
&=\left(H_0^{-1}-H_0^{-1} E_0 H_0^{-1}+\mathcal{O}(\|E_0\|^2)\right)\left(H_1+E_1\right)\\
&=G+H_0^{-1} E_1-H_0^{-1} E_0 H_0^{-1} H_1+\mathcal{O}\left(\|E\|^2\right).
\end{align*}
In the second line, we have used the Taylor expansion for $f(x)=x^{-1}$ and truncated to the first order, assuming $\|E\|$ is small.

Therefore, we have
\begin{equation}
    \Delta G = \widetilde{G} - G = H_0^{-1}(E_1 - E_0G)+\mathcal{O}\left(\|E\|^2\right).
\end{equation}
We then bound the operator norm of by considering the first-order perturbation
\begin{align*}
    \|H_0^{-1}(E_1 - E_0 G)\|&\leq \left\|H_0^{-1}\right\|\|E_1 - E_0 G\|\\
    &\leq\sigma_{\min}^{-1}(H_0)\left(\|E_1\|+\|E_0 G\|\right)\\
    &\leq \sigma_{\min}^{-1}(H_0)\left(\|E\|+\|E\|\|G\|\right)\\
    &\leq \sigma_{\min}^{-1}(H_0)\left(\|E\|+\|E\|\kappa(V)\right).
\end{align*}
Here, the first line follows the sub-multiplicativity of the operator norm. In the second line, we apply $\|H_0^{-1}\|=\sigma_{\min}^{-1}(H_0)$ with $\sigma_{\min}$ denoting the minimal singular value, and the triangle inequality. The third line again follows the sub-multiplicativity. Finally, for $\|G\|$, we apply its eigen-decomposition and get $\|(V^{\rm T})^{-1}ZV^{\rm T}\|\leq \|(V^{\rm T})^{-1}\|\|V^{\rm T}\|\|Z\|\leq \kappa(V)$ such that $\kappa(V)$ is the condition number of $V$ and $\|Z\|\leq 1$ because of the quantum operation is trace non-increasing so that $|z_j|\leq1,~\forall j$.

To see the impact of the change in $G$ for the eigenvalue estimation, we utilize the Bauer-Fike theorem~\cite{bauer1960norms} and get the following results.
\begin{lemma}\label{lem:gevp_dist}
For the generalized eigenvalue problem defined in Definition \ref{def:gevp}, consider that we only have noisy access to the two Hankel matrices as given by Eq.~\eqref{eq:noise_general}. Assume the noisy matrices satisfy $\|E_0\|=\|E_1\|=\|E\|$. Then, the eigenvalues solved by the ideal (denoted as ${\lambda}$) and noisy Hankel matrices (denoted as $\widetilde{\lambda}$) obey
\begin{equation}\label{eq:GEVP_dist}
    \min _{\lambda \in Z}|\widetilde{\lambda}-\lambda| \leq\frac{\|E\|(\kappa^2(V)+\kappa(V))}{\sigma_{\min}^2(V)c_{\min}}.
\end{equation}
\end{lemma}
\begin{proof}
The Bauer-Fike theorem~\cite{bauer1960norms} states that for diagonalizable matrix $T=X\Lambda X^{-1}$ with eigenvalues $\{\eta\}$, the eigenvalues $\{\widetilde\eta\}$ of perturbed matrix $\widetilde{T}=T+\Delta T$ obey
$$
\min _{\eta \in \Lambda}|\widetilde{\eta}-\eta| \leq \kappa(X)\|\Delta T\|.
$$
As such, we can bound the worst-case eigenvalue change in Eq.~\eqref{eq:noise_general} from Eq.~\eqref{eq:generalized_ev} as
\begin{equation*}
    \min _{\lambda \in Z}|\widetilde{\lambda}-\lambda| \leq \kappa(V)\|\Delta G\|\lesssim \kappa(V)\sigma_{\min}^{-1}(H_0)\|E\|\left(1+\kappa(V)\right)\leq \frac{\|E\|(\kappa^2(V)+\kappa(V))}{\sigma_{\min}^2(V)c_{\min}}.
\end{equation*}
Here, the last inequality follows from $\sigma_{\min}^{-1}(H_0)\leq \sigma_{\min}^{-2}(V)c_{\min}^{-1}$ by the Vandermonde decomposition and sub-multiplicativity.
\end{proof}

\subsection{Bounds on the condition number of the Vandermonde matrix}
Our next step is to bound the condition number of $V$. First, we bound the spectral norm of $V$. Since all entries of $V$ have magnitude at most 1 (due to $|z_j| \leq 1$), we apply the bound on the Frobenius norm:
$$\|V\|_F^2 = \sum_{j=1}^r \sum_{k=1}^r |V_{j,k}|^2 \leq r^2.$$
Because the Frobenius norm bounds the operator norm, we have:
\begin{equation}\label{eq:v_op_norm}
    \|V\| \leq \|V\|_F \leq r.
\end{equation}

Next, we find that the upper bound on the $\sigma_{\min}^{-1}(V)$ coincides with $\|V^{-1}\|$. As such, we construct the upper bound by deriving the explicit form of $V^{-1}$. 
Although the form of this inverse is known, we derive it here for completeness. We want to find the form of a matrix $W = (w_{jk})_{j,k=1}^r$ such that $VW=I$. This means for $j\in[r]$, we have 
\begin{equation}\label{eq:interpo}
    \sum_{k=1}^r z_j^{k-1} w_{kl} = \delta_{jl},
\end{equation}
where $\delta_{jl}$ is the Kronecker delta function. Our target then becomes \emph{deriving the form of each coefficient $w_{kl}$ given by the above formula.}

Now, let us notice that for a fixed $l$, Eq.~\eqref{eq:interpo} can be interpreted as a degree-$(r-1)$ polynomial function of the input $z_j$ with coefficients $w_{kl}$. 
As such, each column of $W$ can be viewed as a solution to the polynomial interpolation problem given by Eq.~\eqref{eq:interpo}. To this end, we define the polynomial $p_l(z)$ of degree at most $r-1$ with coefficients correspond to the $l$th column of $W$:
$$p_l(z_j) = \delta_{jl},~j\in[r].$$

A simple solution to the polynomial interpolation problem is by applying the \emph{Lagrange polynomial}. We briefly review the idea of this method. The Lagrange polynomial provides the unique lowest-degree polynomial interpolation for a given set of data points. Consider a set of distinctive input $\{x_i\}_{i=0}^K$ with values $\{y_i\}_{i=0}^K$. Each Lagrange polynomial comprising the Lagrange basis is defined as
$$l_j(x):=\prod_{\substack{0 \leq i \leq k \\ i \neq j}} \frac{x-x_i}{x_j-x_i}.$$
We find the Lagrange polynomial obeys $l_j(x_k)=\delta_{jk}$. Then, by defining the Lagrange interpolation polynomial
$$L(x):=\sum_{j=0}^K y_jl_j(x),$$
we find the polynomial interpolates the data points as $L(x_j)=y_j$.
In our cases, for a fixed $l$, we define its Lagrange interpolation polynomial as 
\begin{equation}\label{eq:interpo2}
    p_l(z) := \sum_{j=1}^r \delta_{jl} \frac{\prod_{k \neq j}(z - z_k)}{\prod_{k \neq j}(z_j - z_k)}=\frac{\prod_{k \neq l}(z - z_k)}{\prod_{k \neq l}(z_l - z_k)}.
\end{equation}
Here, the last formula results from $\delta_{jl} = 0$ for $j \neq l$. We find that by expanding the formula into a summation of monomials, the coefficient for each $z_k$ gives $w_{kl}$. To this end, we resort to the \emph{elementary symmetric polynomial}, which is defined as follows.
\begin{definition}[Elementary symmetric polynomial]
Given variables $x_1, x_2, \ldots, x_n$, the elementary symmetric polynomial $e_m(x_1, x_2, \ldots, x_n)$ is defined as:
$$e_m(x_1, x_2, \ldots, x_n) = \sum_{1 \leq i_1 < i_2 < \ldots < i_m \leq n} x_{i_1} x_{i_2} \cdots x_{i_m}$$
That is, $e_m$ is the summation of all possible products of $m$ distinct variables chosen from the $n$ variables.
\end{definition}


Applying the elementary symmetric polynomial, the nominator of the polynomial given by Eq.~\eqref{eq:interpo2} can be expanded as
$$\prod_{k \neq l}(z - z_k) = \sum_{m=0}^{r-1} (-1)^{r-1-m} e_{r-1-m}(z_1,...,z_{l-1},z_{l+1},...,z_r) \cdot z^m.$$
As such, we obtain each element in $W$ as
$$w_{jk} = \frac{(-1)^{r-j} e_{r-j}(z_1,...,z_{k-1},z_{k+1},...,z_r)}{\prod_{l \neq k}(z_k - z_l)}.$$
This finishes the derivation for the explicit form of $V^{-1}$. Our strategy for bounding $\|V^{-1}\|$ again follows by bounding the Frobenius norm as dictated by Eq.~\eqref{eq:v_op_norm}. We first consider the bound for the elementary symmetric polynomials. Because of $|z_j| \leq 1,~\forall j$, we have
$$|e_m(z_1,...,z_n)| \leq \binom{n}{m},$$
which follows from counting the number of components in the polynomials and iteratively applying the triangle inequality. This readily leads to
$$|w_{jk}| \leq \frac{\binom{r-1}{r-j}}{|\prod_{l \neq k}(z_k - z_l)|}\leq\frac{\binom{r-1}{r-j}}{\Delta^{r-1}},$$
where 
\begin{equation}\label{eq:gap}
    \Delta := \min_{j \neq k} |z_j - z_k|
\end{equation}
is the minimal gap of the signal. Following Eq.~\eqref{eq:v_op_norm}, we have 
$$\|V^{-1}\| \leq \|V^{-1}\|_F = \sqrt{\sum_{j=1}^r \sum_{k=1}^r |w_{jk}|^2}\leq \sqrt{\sum_{j=1}^r \sum_{k=1}^r \frac{\binom{r-1}{r-j}^2}{\Delta^{2(r-1)}}}=\frac{ \sqrt{r\binom{2(r-1)}{r-1}}}{\Delta^{r-1}},$$
where the last equation follows from the identity $\sum_{j=0}^{r-1} \binom{r-1}{j}^2 = \binom{2(r-1)}{r-1}$. Then, for large enough $r$, we use Stirling's approximation for the binomial coefficient such that
$\binom{2(r-1)}{r-1} \approx \frac{4^{r-1}}{\sqrt{\pi(r-1)}}.$
Eventually, we have
\begin{equation}
    \sigma_{\min}^{-1}(V)=\|V^{-1}\| \leq \|V^{-1}\|_F \lesssim \frac{2^{r-1}\sqrt{r}}{\Delta^{r-1} \sqrt{\pi(r-1)/4}} \approx \frac{2^{r}}{\sqrt{\pi}\Delta^{r-1}}.
\end{equation}
Besides, follows Eq.~\eqref{eq:v_op_norm}, we upper bound the condition number as
\begin{equation}\label{eq:kappa_v}
    \kappa(V)=\|V\|\|V^{-1}\|\leq \frac{2^{r} r}{\sqrt{\pi}\Delta^{r-1}}.
\end{equation}

We notice that the bound on the condition number and $\|V^{-1}\|$ seems rather pessimistic: Both quantities scale exponentially with the sparsity $r$, indicating the algorithm remains efficient only when $r=\mathcal{O}(1)$. Yet, we found the scaling to be inevitable as a lower bound of the condition number is exponential in the matrix size, as displayed in Ref.~\cite{pan2016bad}.

We close this section by deriving the error in the generalized eigenvalue problem.
\begin{lemma}\label{lem:sample_complexity}
Consider the generalized eigenvalue problem defined by Definition \ref{def:gevp} with only noisy access to $H_0$ and $H_1$ as given by $\widetilde{H}_0:=H_0+E_0$ and $\widetilde{H}_1:=H_1+E_1$, respectively. Assume the noisy matrices satisfy $\|E_0\|=\|E_1\|=\|E\|$. Then, the generalized eigenvalues solved by the noisy (denoted as $\widetilde{\lambda}$) and ideal Hankel matrices satisfy
\begin{equation}\label{sigma_min}
    \min _{\lambda \in Z}|\widetilde{\lambda}-\lambda| \leq \frac{\|E\|}{c_{\min}}\left( \frac{r^2 2^{4r}}{\pi^2\Delta^{4(r-1)}}+\frac{r 2^{3r}}{\pi^{3/2}\Delta^{3(r-1)}} \right).
\end{equation}
\end{lemma}
\begin{proof}
The result manifests by plugging in the upper bound on $\sigma_{\min}^{-1}$ given by Eq.~\eqref{sigma_min} and the condition number of $V$ given by Eq.~\eqref{eq:kappa_v} into Lemma \ref{lem:gevp_dist}.
\end{proof}

\subsection{Sample complexity from random matrix theory}
This section solves the algorithm's sample complexity by resorting to random matrix theory. The main result we will apply is from Ref.~\cite{tropp2012user} as follows.
\begin{lemma}[Matrix Bernstein for Bounded Random Matrices {\cite[Theorem 1.6]{tropp2012user}}]\label{lemma:matrix_bernstein}
Let $Q_1, Q_2, \ldots, Q_m$ be i.i.d.~random matrices in $\mathbb{C}^{d_1 \times d_2}$ with $\mathbb{E}[Q_i] = 0$ and $\|Q_i\| \leq R$ almost surely for all $i$. Define 
\begin{equation}\label{eq:variance}
    \sigma^2 := \max\left\{\left\|\sum_{i=1}^m \mathbb{E}[Q_i Q_i^\dagger]\right\|, \left\|\sum_{i=1}^m \mathbb{E}[Q_i^\dagger Q_i]\right\|\right\}
\end{equation}
as the variance of the random matrices. Then for any $\gamma > 0$:
\begin{equation}\label{eq:rand_concentration}
\mathbb{P}\left(\left\|\sum_{i=1}^m Q_i\right\|  \geq \gamma\right) \leq (d_1 + d_2) \exp\left(-\frac{\gamma^2/2}{\sigma^2 + R\gamma/3}\right).
\end{equation}
\end{lemma}
The scheme for exploiting is to view the noises as random matrices. Specifically, for each entry of the Hankel matrix, we take $m$ independent samples of both real and imaginary parts. We formulate each sample for different entries into a matrix form and denote entries in the matrix as
\begin{itemize}
    \item $X_{j,k,1}, X_{j,k,2}, ..., X_{j,k,m}$ for the real part;
    \item $Y_{j,k,1}, Y_{j,k,2}, ..., Y_{j,k,m}$ for the imaginary part.
\end{itemize}
Note that the samples are taken from a Hadmard test circuit so that each $X_{j,k,i}$ and $Y_{j,k,i}$ is a random variable taking values in $\{-1, 1\}$. The empirical mean of the random matrices then serves as the unbiased estimator for the ideal Hankel matrix, such that for each entry, we denote it as
\begin{equation}\label{eq:signal_estimator}
\hat{f}(j+k) := \frac{1}{m}\sum_{i=1}^m X_{j,k,i} + i\frac{1}{m}\sum_{i=1}^m Y_{j,k,i}.
\end{equation}
The noise matrix is subsequently given by
\begin{equation}
\begin{aligned}
    E_{j,k} &= \hat{f}(j+k) - f(j+k) \\
&= \frac{1}{m}\sum_{i=1}^m X_{j,k,i} - \text{Re}(f(j+k)) + i\left(\frac{1}{m}\sum_{i=1}^m Y_{j,k,i} - \text{Im}(f(j+k))\right).
\end{aligned}
\end{equation}
The noise matrix can be seen as taking the empirical mean of random matrices $\{Z_i\}_{i=1}^m$ defined as 
\begin{equation}\label{eq:rand_mean}
    [Z_i]_{j,k}:=X_{j,k,i}-\mathbb{E}[X_{j,k,i}] +i \left( Y_{j,k,i} -  \mathbb{E}[Y_{j,k,i}]\right)=X_{j,k,i}' +iY_{j,k,i}',
\end{equation}
where $\frac{1}{m}\sum_{i=1}^m Z_i=E$, $\mathbb{E}[X_{j,k,i}]\equiv\text{Re}(f(j+k))$, $\mathbb{E}[Y_{j,k,i}]\equiv\text{Im}(f(j+k))$, $X_{j,k,i}':=X_{j,k,i}-\mathbb{E}[X_{j,k,i}]$ and $Y_{j,k,i}':=Y_{j,k,i}-\mathbb{E}[Y_{j,k,i}]$. Note that $\mathbb{E}[X_{j,k,i}']=\mathbb{E}[Y_{j,k,i}']=0$.
Because of $\mathbb{E}(E)=0$, it is now clear that we can take the set of $Z_i$ as the subject for Lemma \ref{lemma:matrix_bernstein}. Therefore, the two key components left for solving are an upper bound on the operator norm of $Z_i$ and its variance. 

To derive the upper bound of the operator norm of $Z_i$, we again use the Frobenius norm as the upper bound. For each entry in $Z_i$, we have
$$
\left|[Z_i]_{j,k}\right|=\left|X_{j,k,i}-\mathbb{E}[X_{j,k,i}] +i \left( Y_{j,k,i} -  \mathbb{E}[Y_{j,k,i}]\right)\right|\leq 2\sqrt{2},
$$
where we have used the triangle inequality and $X_{j,k,i},Y_{j,k,i}\in\{-1,1\}$. Therefore, we obtain
\begin{equation}\label{eq:rand_op_norm}
    \|Z_i\|\leq\|Z_i\|_F = \sqrt{\sum_{j=1}^r \sum_{k=1}^r |[Z_i]_{j,k}|^2} \leq 2\sqrt{2} r,~\forall i.
\end{equation}

The next step is to bound $\delta^2$. To this end, let us notice that the entries in the Hankel matrix with the same sum of indices are equal: If $j+k = j'+k'$, then $f(j+k) = f(j'+k')$. Hence, we may reuse the empirical results for repeated entries in the Hankel matrix. By denoting $t = j+k$, we can relabel the real and imaginary parts of the zero-centered random matrix in Eq.~\eqref{eq:rand_mean} as
$$
X'_{t,i},~Y'_{t,i} \text{ for } t \in [2r].
$$
Hence, the matrix $Z_i$ has entries $[Z_i]_{j,k} = X'_{j+k,i} + iY'_{j+k,i}$. Now, consider the target $Z_i Z_i^\dagger$, and its entry is computed as
\begin{equation}\label{eq:z_expand}
\begin{aligned}
    [Z_i Z_i^\dagger]_{j,l} &= \sum_{k=1}^{r} [Z_i]_{j,k} [Z_i^*]_{l,k}\\
    &=\sum_{k=1}^{r} (X'_{j+k,i} + iY'_{j+k,i})(X'_{l+k,i} - iY'_{l+k,i})\\
    &=\sum_{k=1}^{r} (X'_{j+k,i}X'_{l+k,i} + Y'_{j+k,i}Y'_{l+k,i}  + iY'_{j+k,i}X'_{l+k,i} - i X'_{j+k,i}Y'_{l+k,i}),
\end{aligned}
\end{equation}
where $x^*$ denotes complex conjugation of $x$; we substitute in the definition of entries of $Z_i$ in the second line and expand it in the last line. Then, for $\mathbb{E}[Z_i Z_i^\dagger]$, we consider separately the contribution from each term expanded in the last line of Eq.~\eqref{eq:z_expand}. Because of $\mathbb{E}[X'_{t,i}] = \mathbb{E}[Y'_{t,i}] = 0$, we have
\begin{itemize}
    \item $\mathbb{E}[X'_{t,i}X'_{t',i}] = \mathbb{E}[X'_{t,i}]\mathbb{E}[X'_{t',i}] = 0$, $\forall t \neq t'$;
    \item $\mathbb{E}[Y'_{t,i}Y'_{t',i}]=\mathbb{E}[Y'_{t,i}]\mathbb{E}[Y'_{t',i}] = 0$, $\forall t \neq t'$;
    \item $\mathbb{E}[X'_{t,i}Y'_{t',i}]=\mathbb{E}[X'_{t,i}]\mathbb{E}[Y'_{t',i}] = 0$, $\forall t, t'$.
\end{itemize}
The factorization of the above expectation values results simply from the independence of probabilities. This leaves the only non-zero terms to be
$$
\mathbb{E}[(X'_{t,i})^{2}] = \mathbb{E}[(Y'_{t,i})^{2}] \leq1.
$$
The rationale for the final upper bound follows from $X_{t,i}, Y_{t,i}\in\{-1,1\}$, $X'_{t,i}=X_{t,i}-\mathbb{E}[X_{t,i}]~\text{and}~Y'_{t,i}=Y_{t,i}-\mathbb{E}[Y_{t,i}]$ so that $\mathbb{E}[(X'_{t,i})^{2}]=\mathbb{E}[(X_{t,i})^{2}]+\mathbb{E}^2[X_{t,i}]-2\mathbb{E}^2[X_{t,i}]=1-\mathbb{E}^2[X_{t,i}]\leq1$ and similarly for $\mathbb{E}[(Y'_{t,i})^{2}]$. We finally arrive at
$$
\mathbb{E}[Z_i Z_i^\dagger]_{j,l} \leq \sum_{k=1}^{r} (\delta_{j+k,l+k} + \delta_{j+k,l+k})=2\sum_{k=1}^{r} \delta_{j,l},
$$
suggesting the matrix is diagonal and each diagonal element is bounded by $2r$. As such, we can compute the variance given by Eq.~\eqref{eq:variance} as
$$
\left\|\sum_{i=1}^m \mathbb{E}[Z_i Z_i^\dagger]\right\| \leq 2rm.
$$
Similarly, we can show $\left\|\sum_{i=1}^m \mathbb{E}[Z_i^\dagger Z_i]\right\| \leq 2rm$. Eventually, we have 
\begin{equation}\label{eq:rand_variance}
    \sigma^2\leq 2rm.
\end{equation}

At this point, we obtain bounds on the operator norm for the noise matrix in the following lemma.
\begin{lemma}\label{lem:rand_op_norm}
Let $\varepsilon,\delta\in (0,1)$. Let $Z_1, Z_2, \ldots, Z_m$ be i.i.d.~random matrices in $\mathbb{C}^{r \times r}$ with element given by Eq.~\eqref{eq:rand_mean}. Then, the operator norm of the mean of the random matrices satisfies
$$
\mathbb{P}\left(\left\|\frac{1}{m}\sum_{i=1}^m Z_i\right\|  \geq \varepsilon\right) \leq \delta
$$
for the sample complexity to be
\begin{equation}\label{eq:rand_samp_complexity}
    m=\mathcal{O}\left(\frac{r}{\varepsilon^2}  \log\left(\frac{r}{\delta}\right)\right).
\end{equation}
\end{lemma}
\begin{proof}
We remark that Lemma \ref{lemma:matrix_bernstein} characterizes the behavior of the sum of random matrices. However, we are interested in the behavior of their mean value. Because we aim to upper bound the operator norm of $E=\frac{1}{m}\sum_{i=1}^m Z_i$ by $\varepsilon$, this is equivalent to bound the sum of $Z_i$ by $m\varepsilon$ using ~Eq.~\eqref{eq:rand_op_norm}:
$$
\mathbb{P}\left(\left\|\sum_{i=1}^m Z_i\right\|  \geq m\varepsilon\right) \leq 2r \exp\left(-\frac{(m\varepsilon)^2/2}{2rm + 2\sqrt{2}r(m\varepsilon)/3}\right),
$$
where we have taken $d_1=d_2=r$, $\sigma^2=2rm$ from Eq.~\eqref{eq:rand_variance} and $R=2\sqrt{2}r$ from Eq.~\eqref{eq:rand_op_norm}. 

Equivalently, we have
$$
\mathbb{P}\left(\left\| E\right\|  \geq \varepsilon\right) \leq 2r \exp\left(-\frac{(m\varepsilon)^2/2}{2rm + 2\sqrt{2}r(m\varepsilon)/3}\right)\leq \delta.
$$
The $2rm$ term dominates the denominator in the exponential function, so we will take the denominator to be $\mathcal{O}(rm)$.
To fulfill the inequality, we obtain the sample complexity as given by Eq.~\eqref{eq:rand_samp_complexity}.
\end{proof}

Finally, we are in place for deriving the sample complexity of the GEVP problem given by Definition \ref{def:gevp}.
\begin{lemma}\label{lem:final_sample_complexity}
Let $\varepsilon_1,\delta_1\in (0,1)$. Let $H_0$ and $H_1$ be the Hankel matrices in Definition \ref{def:gevp} with each element a signal of the form given by Eq.~\eqref{eq:signal}. Let $\widetilde{H}_0$ and $\widetilde{H}_1$ be the noisy version of $H_0$ and $H_1$ such that each element (real and imaginary parts) is estimated through the empirical mean of $m$ samples of Bernoulli variables as shown by Eq.~\eqref{eq:signal_estimator}, provided the estimator to be unbiased. Then, to let the generalized eigenvalues solved by the noisy (denoted as $\widetilde{\lambda}$) and ideal Hankel matrices obey
$$
\min _{\lambda \in Z}|\widetilde{\lambda}-\lambda|\leq \varepsilon_1
$$
with probability at least $1-\delta_1$, the sample complexity is then given by
\begin{equation}\label{eq:sample_complexity_general}
    m=\mathcal{O}\left(\frac{r^5 2^{8r}}{ (\Delta^{4r}\varepsilon_1 c_{\min})^2} \log\left(\frac{r}{\delta_1}\right)\right).
\end{equation}
\end{lemma}
\begin{proof}
To accomplish the accuracy $\varepsilon_1$, it requires the operator norm of the noise matrix to abide by
\begin{equation}\label{eq:noise_prop}
    \|E\|\leq\frac{c_{\min} \varepsilon_1}{F}
\end{equation}
with $F=\frac{r^2 2^{4r}}{\pi^2\Delta^{4r}}+\frac{r 2^{3r}}{\pi^{3/2}\Delta^{3r}}$, as dictated by Lemma \ref{lem:sample_complexity}. 
All $r-1$ is taken to be $r$ for simplicity.
Next, taking this operator norm and $\delta=\delta_1$ into Lemma \ref{lem:rand_op_norm}, we have the sample complexity to be
$$\mathcal{O}\left(\frac{r F^2}{c_{\min }^2 \varepsilon_1^2} \log \frac{r}{\delta_1}\right).$$
Finally, we set $F=\mathcal{O}\left((2/\Delta)^{4r})r^2\right)$ for $\Delta<1$.
Note also that there are $\mathcal{O}(r)$ different signals.
This concludes the proof.
\end{proof}

\subsection{Complexity for solving the QESE problem}
In this section, we address the sample, query, and total complexity of our algorithm for solving the QESE problem. 

Specifically, we consider two types of (ideal) signals as shown by Eq.~\eqref{eq:signal}. The first type is that the signal is generated by monomials of the matrix $A$ as shown in Eq.~\eqref{eq:q_signal}. Beyond this, we also consider signals that are generated by other matrix functions, as Eq.~\eqref{eq:q_signal_f}. Particularly, we consider the matrix exponential that is of practical importance, as will later be discussed in Sec.~\ref{sec:app}.

Before diving into the details of the deductions, we first summarize the steps of our algorithm for solving the QESE problem for the second type of signals.
\begin{enumerate}
    \item Prepare the initial density matrix $\rho_I$ or its purification $\ket{{\rho_I}}$ if the purified quantum query access is available.
    \item Construct the Hadamard test circuit shown by Eq.~\eqref{eq:circ_hadamard_test} for the real part of the signal with $W=I$. For the imaginary part, change to $W=S^\dagger$.
    \item When $\ket{\rho_I}$ is accessible or $\rho_I$ is pure, perform QAE for estimating the real and imaginary parts of the signal through Eq.~\eqref{eq:ae_real} and \eqref{eq:ae_img} with sufficient accuracy. Otherwise, proceed by standard Hadamard test procedures with estimators of the real and imaginary parts of the signal given by Eq.~\eqref{eq:hadamard_estimator1} or \eqref{eq:hadamard_estimator2} depending on the matrix function applied.
    \item Repeat for each signal in the Hankel matrix the above three steps by constructing block encodings of the corresponding matrix function (with the QEVT algorithm if necessary).
    \item Collecting the two noisy Hankel matrices $\widetilde{H}_0$ and $\widetilde{H}_1$. Solve for the eigenvalues of the GEVP given by Eq.~\eqref{eq:noise_general} using the two Hankel matrices with the MP method.
\end{enumerate}
For the matrix monomial function, the only difference is that we do not need to invoke the QEVT algorithm, but only by taking products of the BE. Given $U_A$ a ($\alpha_A,m,0$)-BE of $A$, we may change the $b$ register of the Hadamard test circuit used above to $m$ qubits and remove the action of $U_s$. Besides, the final targeted state of register $b$ becomes $\ket{0^m}$, for which the QAE algorithm should change accordingly.

\subsubsection{Matrices with general eigenvalues}\label{sec:general_solver}
We first consider non-normal matrices with general eigenvalues, either real or complex. The matrix function we are concerned with is the monomial functions as shown in Eq.~\eqref{eq:q_signal}.

The sample, query, and total complexities of our algorithm are given as follows, without purified quantum query access or $\rho_I$ is mixed.
\begin{theorem}\label{thm:general_without}
Let $A,\rho_I\in\mathbb{C}^{N\times N}$ be two matrices that satisfy the conditions in Assumption \ref{assm:sparsity} such that $\rho_I$ is $r$-sparse when expanded in the eigenbasis of $A$. Let $\epsilon,\delta\in(0,1)$. Denote the set of $r$ eigenvalues as $Z:=\{\lambda_i\}_{i=1}^r$. Then, given access to the ($\alpha_A,m,0$)-block encoding $U_A$ of $A$, there exists a quantum algorithm that outputs an estimation of these eigenvalues $\{\widetilde\lambda_i\}_{i=1}^r$ such that
$$\min _{\lambda \in Z}|\widetilde{\lambda}-\lambda|\leq \epsilon$$
with success probability no less than $1-\delta$ at the cost
\begin{itemize}
    \item A maximal $\mathcal{O}\left(r\right)$ queries to controlled-$U_A$ and its inverse in one coherent run of the algorithm.
    \item A sample complexity of $\mathcal{O}\left(\frac{r^5  \alpha_A^{4r}2^{8r}}{ ((\Delta')^{4r}\epsilon c_{\min})^2} \log\left(\frac{2r}{\delta}\right)\right)$ for the preparation of $\rho_I$, where 
    \begin{equation}\label{eq:gap_normalized}
        \Delta' := \min_{j \neq k}\frac{1}{\alpha_A} |\lambda_j - \lambda_k|,\lambda\in Z
    \end{equation}
    is the normalized gap and $c_{\min}$ is given by Eq.~\eqref{eq:c_min} in Assumption \ref{assm:sparsity}.
    \item A total number of $\mathcal{O}\left(\frac{r^6  \alpha_A^{4r}2^{8r}}{ ((\Delta')^{4r}\epsilon c_{\min})^2} \log\left(\frac{2r}{\delta}\right)\right)$ queries to controlled-$U_A$ and its inverse.
\end{itemize}
\end{theorem}
\begin{proof}
We first note that the largest degree of monomial functions is $\mathcal{O}(r)$ as dictated by the structure of the Hankel matrices, which determines the query complexity of $U_A$ for constructing the monomial functions at one coherent run. 

The major cost for one coherent run is from the Hadamard test protocol. To this end, we resort to the result given by Eq.~\eqref{eq:sample_complexity_general} of Lemma \ref{lem:final_sample_complexity}. In the light of this result, we take $\varepsilon_1$ in Eq.~\eqref{eq:sample_complexity_general} to be $\varepsilon_1=\epsilon/(\alpha_A)^{-2r}$. This is because we need to multiply back the normalization factor in the monomial function as given by Eq.~\eqref{eq:q_signal} as discussed in Sec.~\ref{sec:q_alg}. Hence, the error is enlarged by a factor at most $(\alpha_A)^{-2r}$ for the highest degree monomial. 
This gives the sample complexity to be $\mathcal{O}\left(\frac{r^5  \alpha_A^{4r}2^{8r}}{ ((\Delta')^{4r}\epsilon c_{\min})^2} \log\left(\frac{2r}{\delta}\right)\right)$ due to $\Delta'<1$.

The total query complexity is obtained by multiplying the query complexity in one run by the sample complexity.

\end{proof}


When given access to the purified quantum query access of the initial density matrix or $\rho_I$ is pure, we have the complexity of the algorithm as follows.
\begin{theorem}\label{thm:general_with}
When given purified query access to $\rho_I$ defined in Definition \ref{def:purified} such that $U_I\ket{0}=\ket{\rho_I}$, the complexities becomes:
\begin{itemize}
    \item A maximal $\mathcal{O}\left(\frac{r^3\alpha_A^{2r}2^{4r}}{\Delta'^{4r}\epsilon c_{\min}} \log\left(\frac{r}{\delta}\right)\right)$ queries to controlled-$U_A$, $U_I$ and their inverse in one coherent run of the algorithm. 
    \item A sample complexity of $\mathcal{O}(r)$.
    \item A total number of $\mathcal{O}\left(\frac{r^4 2^{4r}\alpha_A^{2r}}{\Delta'^{4r}\epsilon c_{\min}} \log\left(\frac{r}{\delta}\right)\right)$ queries to controlled-$U_A$ and $U_I$ and its inverse.
\end{itemize}
\end{theorem}
\begin{proof}
In cases $\ket{\rho_I}$ is accessible, we apply the QAE algorithm for estimating the real and imaginary parts of the signal through estimators given by Eq.~\eqref{eq:ae_real} and \eqref{eq:ae_img}. Note that when the matrix function is a monomial, the signal is normalized with $\alpha_A^k$, where $k$ is the degree of the monomial. Then, according to Lemma \ref{lemma:ae}, $\mathcal{O}\left(\frac{\alpha_A^{2r-1}}{\varepsilon_1}\log(1/\delta_1)\right)$ queries to c$-U_A$ and its inverse is required to guarantee an $\varepsilon_1$ accuracy approximation for each element in the Hankel matrix with probability no less than $1-\delta_1$. Note that another $\mathcal{O}(r)$ queries to $U_A$ and its inverse are required for implementing the product of BE (monomial function). To determine $\varepsilon_1$, using the Frobenius norm as an upper bound on the operator norm, we have the noise matrix of operator norm bounded by $\|E\|\leq r\varepsilon_1$. Then, we invoke Lemma \ref{lem:sample_complexity} and let
$$\frac{\|E\|}{c_{\min}}\left( \frac{r^2 2^{4r}}{\pi^2\Delta'^{4(r-1)}}+\frac{r 2^{3r}}{\pi^{3/2}\Delta'^{3(r-1)}} \right)\leq \frac{r\varepsilon_1}{c_{\min}}\left( \frac{r^2 2^{4r}}{\pi^2\Delta'^{4(r-1)}}+\frac{r 2^{3r}}{\pi^{3/2}\Delta'^{3(r-1)}} \right)\leq\epsilon.$$
This gives $\varepsilon_1=\mathcal{O}(\epsilon c_{\min}r^{-3}(\Delta'/2)^{4(r-1)})$. This determines the query complexity in one coherent run. Besides, we set $\delta_1=\delta/(2r-1)$ to guarantee the union over all $2r-1$ signals with the desirable success probability. 

For the sample complexity, $2r-1$ signals are accounted for in the Hankel matrix. The total complexity is given by multiplying the maximal query and sample complexity. Finally, we modify all $r-1$ to $r$ to keep the notation simple.
\end{proof}

We remark that with the purified density matrix, the total complexity is (approximately) quadratically better in all dependencies. This makes the purified model more powerful when a deeper quantum circuit is available. Besides, we find the algorithm reaches the Heisenberg-limited performance regarding dependence on the accuracy, i.e.~$\mathcal{O}(\epsilon^{-1})$, suggesting its optimality~\cite{giovannetti2004quantum,giovannetti2006quantum}.

\subsubsection{Matrices with real eigenvalues}\label{sec:real_solver}
We consider non-normal matrices with real eigenvalues. In such circumstances, the Chebyshev polynomial of the first kind provides the optimal performance for the approximation of the matrix functions. For the monomial function, the Chebyshev polynomial can approximate the monomial function $x^k$ with degree ${\mathcal{O}}(\sqrt{k\log(\epsilon^{-1})})$~\cite{sachdeva2014faster} of approximation error $\epsilon$. Yet, as the GEVP algorithm is very sensitive to errors in the Hankel matrix, the accuracy grows exponentially with $r$, which could cause a slowdown compared to taking products of BEs. To this end, we instead consider the matrix exponentials for the normalized matrix $A/\alpha_A$
\begin{equation}\label{eq:matrix_exp_r}
    f(x)=e^{-2\pi ixt}.
\end{equation}
Here, the evolution time $t$ takes values from $0,\ldots,2r-1$ for different signals in the Hankel matrix. As such, the signal is given by Eq.~\eqref{eq:q_signal_f}. We find the setting to be intriguing as it reduces the ideal signal to be 
\begin{equation}\label{eq:real_signal}
    f(t)=\sum_{j=1}^r c_j e^{-2\pi i\lambda_j t/\alpha_A},~\lambda_j\in\mathbb{R}
\end{equation}
which is the very problem solved by Prony's method as given by Eq.~\eqref{eq:signal}. This problem is also known as the super-resolution~\cite{donoho1992superresolution,candes2013super} or the sparse Fourier analysis problem~\cite{hassanieh2012nearly} in the signal processing community. The problem, by its own right, has several fascinating features. The most important one is that when using the GEVP solver for addressing the problem, one finds that the dependence on $r$ can be improved exponentially, making $r=\mathcal{O}(\rm{poly}(n))$ a feasible setting. Such a reduction can be understood from the (re-normalized) columns of the Vandermonde matrix becoming orthogonal as $r$ goes to infinity, resulting in the condition number of the Vandermonde matrix being well-behaved (scales polynomially with $r$). To this end, we follow the seminal work of Moitra \cite{moitra2015super} and get the following result.
\begin{lemma}\label{lem:real_gevp_error}
Let $\varepsilon_1,\delta_1\in (0,1)$. Let $H_0$ and $H_1$ be the Hankel matrices in Definition \ref{def:gevp} with each element a signal of the form given by Eq.~\eqref{eq:real_signal}. Let $\widetilde{H}_0$ and $\widetilde{H}_1$ be the noisy versions of $H_0$ and $H_1$ such that $\|\widetilde{H}_0-H_0\|=\|\widetilde{H}_1-H_1\|=\|E\|$. Then, the output of the generalized eigenvalue solver satisfies
\begin{equation}\label{eq:real_noise_prop}
    \min _{\lambda \in Z}|\widetilde{\lambda}-\lambda|\leq \frac{\|E\|(\kappa^4(V)+\kappa^3(V))}{rc_{\min}}
\end{equation}
such that 
\begin{equation}\label{eq:real_kappa_v}
    \kappa(V)^2 \leq \frac{r+1 / \Delta_w-1}{r-1 / \Delta_w-1}
\end{equation}
assuming $r>\frac{1}{\Delta_w}+1$, where $\Delta_w$ is the gap of the signal $\{\lambda_j\}$ under the wrap-around metric, defined as $d(\theta, \phi)=\min _{k \in \mathbb{Z}}|\theta-\phi+2 \pi k|$.
\end{lemma}
\begin{proof}
We use the result shown in Lemma \ref{lem:gevp_dist}, such that
$\min _{\lambda \in Z}|\widetilde{\lambda}-\lambda| \leq \frac{\|E\|(\kappa^2(V)+\kappa(V))}{\sigma_{\min}^2(V)c_{\min}}.$ Under circumstances that the Vandermonde matrix consists of Fourier modes, a tight bound on the condition number is given by {\cite[Theorem 2.3]{moitra2015super}} as shown in Eq.~\eqref{eq:kappa_v}. Besides, it has that $\sigma_{\min }^{-1}(V)=\frac{\kappa(V)}{\sigma_{\max }(V)},$
where $\sigma_{\max }(V)$ is the maximal singular value of $V$ or the spectral norm. As the spectral norm of a matrix is no less than the norm of any of its columns, we can lower bound it as 
$$\sigma_{\max}\geq \|v\|_2=\sqrt{\sum_{j=1}^r\left|z_j^{k}\right|^2}=\sqrt{r}$$
for any $k$ with $v$ any column of $V$ with $\|\cdot\|_2$ the $L_2$ norm. This gives
$$\sigma_{\min }^{-1}(V)\leq\frac{\kappa(V)}{\sqrt{r}},$$
which determines the final result.
\end{proof}

We remark that $\kappa(V)$ remains reasonable when the denominator in Eq.~\eqref{eq:kappa_v} is not exponentially small. This could be the case, such as when $r$ is a larger polynomial of $n$ than the inverse of $\Delta_w$. This highlights the necessity of the separation in the signal, as discussed in {\cite[Sec.3]{moitra2015super}}.

The idea for solving the QESE problem consists of first estimating the normalized eigenvalue of $A/\alpha_A$ and then multiplying back the normalization factor $\alpha_A$ to obtain the final result.
\begin{theorem}[Real eigenvalues without purified quantum query access]\label{thm:real_without}
Suppose the eigen-spectrum of $A$ is real. Denote the Jordan condition number of $A$ as $\kappa_J$ and $d_{\max}$ the largest size of the Jordan block of $A$. Then, given access to the ($\alpha_A,m,0$)-block encoding $U_A$ of $A$ such that $\alpha_A\geq2\|A\|$, there exists a quantum algorithm that outputs an estimation of these eigenvalues $\{\widetilde\lambda_i\}_{i=1}^r$ that succumbs to
$$\min _{\lambda \in Z}|\widetilde{\lambda}-\lambda|\leq \epsilon$$
with success probability no less than $1-\delta$ at the cost
\begin{itemize}
    \item A maximal $\mathcal{O}\left(d^{3/2}\kappa_J^2\log\left(\frac{r\sqrt{d}\kappa_J\alpha_A\kappa(V)}{c_{\min}\epsilon}\right)\log\left(\frac{r\alpha_A\kappa(V)}{c_{\min}\epsilon}\right)\right)$ queries to controlled $U_A$, $U_s$ defined in Eq.~\eqref{eq:Us_controlled_prep} and their inverse in one coherent run of the algorithm, where $d=\mathcal{O}\left(r+\log\left(\frac{\kappa_J\kappa(V)\alpha_A}{c_{\min}\epsilon}\right)\right)$, where $\kappa(V)$ is given by \eqref{eq:real_kappa_v}.
    \item A sample complexity of $\mathcal{O}\left(\frac{\kappa^8(V)\alpha_A^{2}}{c_{\min}^{2}\epsilon^{2}}\log\left(\frac{r}{\delta}\right)\right)$.
    \item A total number of $\mathcal{O}\left(\frac{\kappa^8(V)d^{3/2}\kappa_J^2\alpha_A^{2}}{c_{\min}^{2}\epsilon^{2}}\log\left(\frac{r\sqrt{d}\kappa_J\alpha_A\kappa(V)}{c_{\min}\epsilon}\right)\log\left(\frac{r\alpha_A\kappa(V)}{c_{\min}\epsilon}\right)\log\left(\frac{r}{\delta}\right)\right)$ queries to controlled $U_A$, $U_s$ and its inverse.
\end{itemize}
\end{theorem}
\begin{proof}
The algorithm applies QEVT to approximate the matrix exponentially. The errors in the signals of the Hankel matrix come from two sources: (i) approximation error to the matrix exponential using the QEVT algorithm, and (ii) sampling noise from the Hadamard test circuit. Denoting $H_0'$ and $H_1'$ the Hankel matrices that only contain approximation error from QEVT but not sampling error. Then, to let
$$\|H_i-\widetilde{H}_i\|=\|H_i-H_i'+H_i'-\widetilde{H}_i\|\leq\|H_i-H'_i\|+\|\widetilde{H}_i-H'_i\|\leq\varepsilon_1$$
with $i \in\{0,1\}$, we may set the errors in each of the two sources to be both $\varepsilon_1/2$. 

We first determine $\varepsilon_1$. The output of the GEVP solver is a $(\varepsilon_1/2)$-estimation to $\{z_j:=e^{-2\pi i\lambda_j/\alpha_A}\}_{j=1}^r$. We define $\theta_j:=-\frac{2 \pi \lambda_j}{\alpha_A}$ so that $e^{-2\pi i\lambda_j/\alpha_A}=e^{i\theta_j}$ for some $j$. For our algorithm, if the output (denoted as $\widetilde{\theta}_j$) of the GEVP solver is not a phase factor, we project it onto the closest one on the complex unit circle by normalization. Algorithmically, that is $\widetilde{\theta}_j:=\arg(\widetilde{z}_j)$, where $\widetilde{z}_j:=e^{i\widetilde{\theta}_j}$ is the (normalized) output of the GEVP solver. The final estimation is given by $\widetilde{\lambda}_j=\frac{\alpha_A\widetilde{\theta}_j}{2\pi}$. We use the geometric relationship on the unit circle such that
$$\left|e^{i \bar{\theta}_j}-e^{i \theta_j}\right|=2 \sin \left(\frac{\left|\widetilde{\theta}_j-\theta_j\right|}{2}\right) \leq\varepsilon_1.$$
We have
$$\left|\arg \left(\widetilde{z}_j\right)-\arg \left(z_j\right)\right| \leq 2 \arcsin \left(\frac{\varepsilon_1}{2}\right),$$
resulting in
$$\left|\widetilde{\lambda}_j-\lambda_j\right|=\frac{\alpha_A}{2 \pi}\left|\arg \left(\widetilde{z}_j\right)-\arg \left(z_j\right)\right| \leq \frac{\alpha_A}{\pi} \arcsin \left(\frac{\varepsilon_1}{2}\right)\approx\frac{\alpha_A\varepsilon_1}{2\pi}.$$
The final approximate equality follows from the error being small.
Thus, set $\varepsilon_1=\frac{2 \pi \epsilon}{\alpha_A}$ gives the result $|\widetilde{\lambda}_j-\lambda_j|\leq \epsilon$. 

For the approximation error, we invoke Lemma \ref{lem:qevt_r_block}; specifically, Eq.~\eqref{eq:qevt_be_complexity_r}. We first determine the truncation order for approximating the exponential function in Eq.~\eqref{eq:matrix_exp_r} using the Chebyshev polynomial.
We denote the final truncation error to be $\varepsilon_2$, to be determined later.
The scaling of the approximation error scaling with the truncation order is given by \cite[Lemma 23]{low2024quantum}, where a degree-$(d-1)$ polynomial has error scales as $\mathcal{O}\left(\kappa_J\left(\frac{e d_{\max } \tau}{2 d}\right)^d\right)$ for $\alpha_A\geq2\|A\|$. 
This gives us $d=\mathcal{O}(d_{\max}r+\log(\kappa_J/\varepsilon_2))$. The maximal evolution time is $(2r-1)$ for the normalized operator $A/\alpha_A$. For the $\|p(\cos \theta) \sin \theta\|_{2|-\pi, \pi|}$ term in Eq.~\eqref{eq:qevt_be_complexity_r}, we have that
$$\|p(\cos \theta) \sin \theta\|_{2|-\pi, \pi|}=\left(\int_{-\pi}^\pi|p(\cos \theta) \sin \theta|^2 d \theta\right)^{1 / 2}\leq |p(\cos \theta)|\left(\int_{-\pi}^\pi| \sin \theta|^2 d \theta\right)^{1 / 2}=\sqrt{\pi}|p(\cos \theta)|.$$
Besides, the polynomial function satisfies $|p(x)/(2\alpha_p)|\leq1,~x\in[-1,1]$, indicating the matrix polynomial is trace non-increasing. This gives us: $\frac{\|p(\cos \theta) \sin \theta\|_{2|-\pi, \pi|}}{\alpha_p}=\mathcal{O}(1)$.

For determine $\varepsilon_2$, we need $\|H_0-H'_0\|=\frac{\varepsilon_1}{2}=\frac{\pi \epsilon}{\alpha_A}$. As each entry in the Hankel matrix has a $\varepsilon_2$ error induced by the polynomial approximation, the operator norm of the noise matrix is bounded by the Frobenius norm, which is $r\varepsilon_2$.
Subsequently, by Lemma \ref{lem:real_gevp_error}, the noise matrix has operator norm obey
\begin{equation}
    \frac{\|E\|(\kappa^4(V)+\kappa^3(V))}{rc_{\min}}\leq \frac{r\varepsilon_2(\kappa^4(V)+\kappa^3(V))}{rc_{\min}}\leq\frac{\varepsilon_1}{2}=\frac{\pi \epsilon}{\alpha_A}.
\end{equation}
This gives $\varepsilon_2=\mathcal{O}\left(\frac{c_{\min}\epsilon}{\alpha_A\kappa^4(V)}\right)$. This concludes the truncation order $d$.
Also, according to Eq.~\eqref{eq:qevt_be_complexity_r}, we have the maximal query complexity as $\mathcal{O}\left(d^{3/2}\alphaU^2\log\left(\frac{r\sqrt{d}\alphaU}{\varepsilon_2}\right)\log\left(\frac{r}{\varepsilon_2}\right)\right)=\mathcal{O}\left(d^{3/2}\alphaU^2\log\left(\frac{r\sqrt{d}\alphaU\alpha_A\kappa(V)}{c_{\min}\epsilon}\right)\log\left(\frac{r\alpha_A\kappa(V)}{c_{\min}\epsilon}\right)\right)$.

For the sampling error, by Lemma \ref{lem:real_gevp_error}, the noise matrix should have operator norm obey
$$\frac{\|E\|(\kappa^4(V)+\kappa^3(V))}{rc_{\min}}\leq\frac{\varepsilon_1}{2}=\frac{\pi \epsilon}{\alpha_A}.$$
It is reasonable to assume that $\Delta_w$ in Eq.~\eqref{eq:real_kappa_v} is small so that the error propagation given by Eq.~\eqref{eq:real_noise_prop} is leading by $\kappa^4(V)$. 
This determines $\|E\|=\mathcal{O}\left(\frac{r\epsilon c_{\min}}{\kappa^4(V)\alpha_A}\right)$. As such, we invoke Lemma \ref{lem:rand_op_norm} to bound the operator norm of the noise matrix $E$. To guarantee a failure probability no more than $\delta$, we have the sample complexity to be $\mathcal{O}\left(\frac{\kappa^8(V)\alpha_A^2}{r\epsilon^{2}c_{\min}^{2}}\log\left(\frac{r}{\delta}\right)\right)$ for estimating a single signal. Besides, there are $\mathcal{O}(r)$ different signals in the Hankel matrix.

Finally, $\alphaU=\mathcal{O}\left(\kappa_J d^{d_{\max }-1}\right)$~ \cite[Appendix A.2]{low2024quantum}. Thus, when the matrix becomes diagonalizable, one finds that $d_{\max}=1$ and $\alphaU=\mathcal{O}(\kappa_J)$.

This concludes the three claimed complexities.
\end{proof}

When the initial density matrix is a pure state, or we have purified quantum query access to it, we have the following result.
\begin{theorem}[Real eigenvalues with purified quantum query access]\label{thm:real_with}
Suppose the eigen-spectrum of $A$ is real. When given purified quantum query access, the costs become:
\begin{itemize}
    \item A maximal $\mathcal{O}\left(\frac{\alpha_A \kappa^4(V)d^{3/2}\kappa_J^2}{c_{\min}\epsilon}\log\left(\frac{r\sqrt{d}\kappa_J\alpha_A\kappa(V)}{c_{\min}\epsilon}\right)\log\left(\frac{r\alpha_A\kappa(V)}{c_{\min}\epsilon}\right)\log\left(\frac{r}{\delta} \right)\right)$ queries to controlled-$U_A$ and its inverse in one coherent run of the algorithm, where $d=\mathcal{O}\left(r+\log\left(\frac{\kappa_J\kappa(V)\alpha_A}{c_{\min}\epsilon}\right)\right)$, and $\alphaU$ satisfies Eq.~\eqref{eq:alphaU}.
    \item A sample complexity of $\mathcal{O}\left(r\right)$.
    \item A total number of $\mathcal{O}\left(\frac{\alpha_A r\kappa^4(V)d^{3/2}\kappa_J^2}{c_{\min}\epsilon}\log\left(\frac{r\sqrt{d}\kappa_J\alpha_A\kappa(V)}{c_{\min}\epsilon}\right)\log\left(\frac{r\alpha_A\kappa(V)}{c_{\min}\epsilon}\right)\log\left(\frac{r}{\delta} \right)\right)$ queries to controlled-$U_A$ and its inverse.
\end{itemize}
\end{theorem}
\begin{proof}
The analysis is similar to that of Theorem \ref{thm:real_without}. The only change is in introducing the QAE such that the error $\varepsilon_1$ induced by it for each signal should obey
$$\frac{\|E\|(\kappa^4(V)+\kappa^3(V))}{rc_{\min}}\leq\frac{r\varepsilon_1(\kappa^4(V)+\kappa^3(V))}{rc_{\min}}\leq\frac{\pi \epsilon}{\alpha_A},$$
where we have used the Frobenius norm as an upper bound to the operator norm.
Thus, we have $\varepsilon_1=\Theta\left(\frac{c_{\min}\epsilon}{\alpha_A \kappa^4(V)}\right)$. Then, applying Lemma \ref{lemma:ae} gives us the final results.

\end{proof}

\subsubsection{Matrices with complex eigenvalues}\label{sec:complex_solver}
The analysis of matrices with complex eigenvalues is parallel to the real case. Yet, we apply a different matrix exponential function, which is given by 
\begin{equation}
    f(A/\alpha_A)=e^{At/\alpha_A},\quad t=0,\ldots,2j-1.
\end{equation}
To this end, we assume the eigenvalues of the matrix $A$ satisfy
\begin{equation}\label{eq:matrix_assm}
    \text{Re}(\lambda)\leq0
\end{equation}
so that $f(A/\alpha_A)$ is trace non-increasing. This matrix function arises naturally in practical scenarios, such as open quantum dynamics and ordinary differential equations. Our strategy is similar to the real cases, where we estimate the eigenvalue of the matrix exponential using the GEVP solver and then output the estimation to the original eigenvalues through post-processing. A primary reason to consider this method over the general one in Sec.~\ref{sec:general_solver} is its ability to detect the $\lambda=0$ eigenvalue, a case the general method misses. This capability might be attributed to its signal produced, analogous to Prony's method, which allows it to handle zero eigenvalues.

\begin{theorem}[Complex eigenvalues without purified quantum query access]\label{thm:complex_without}
For $A$ with complex eigenvalues, suppose that all eigenvalues of $A$ satisfy $\rm{Re}(\lambda(A))\leq0$. Given access to the ($\alpha_A,m,0$)-block encoding $U_A$ of $A$ such that $\alpha_A\geq \|A\|$, suppose that the numerical range $\mathcal{W}\left(A / \alpha_A\right)$ is encompassed by a Faber region $\mathcal{E}$, which is convex and symmetric
for the left half real axis on the complex plane with corresponding conformal maps $\mathcal{E}^c \rightarrow \mathcal{D}^c, {\Psi}: \mathcal{D}^c \rightarrow \mathcal{E}^c$ for the Faber polynomial $F_d(x)$.

Then, there exists a quantum algorithm that outputs an estimation of these eigenvalues $\{\widetilde\lambda_i\}_{i=1}^r$ that succumbs to
$$\min _{\lambda \in Z}|\widetilde{\lambda}-\lambda|\leq \epsilon$$
with success probability no less than $1-\delta$ at the cost
\begin{itemize}
    \item A maximal ${\mathcal{O}\left(\alpha_\beta d^{3/2}\alphaFP ^2\log\left(\frac{\alpha_A\alpha_\beta\sqrt{d}\alphaFP r^3 2^{4r}}{\epsilon c_{\min}(\Delta')^{4r}}\right)\log\left(\frac{\alpha_A r^3 2^{4r}}{\epsilon c_{\min}(\Delta')^{4r}}\right)\right)}$ queries to controlled $U_A$, $U_s$ defined in Eq.~\eqref{eq:Us_controlled_prep} and their inverse in one coherent run of the algorithm, where $d=\mathcal{O}\left(e^\zeta r+\log\left(\frac{r^3 2^{4r}}{\epsilon c_{\min}(\Delta')^{4r}}\right)\right)$ with $\zeta:=\Psi'(\infty)$ is the capacity of the Faber region, $\Delta'$ is the normalized gap as given by Eq.~\eqref{eq:gap_normalized} and $\alpha_\beta$ defined by Eq.~\eqref{eq:alpha beta}.
    \item A sample complexity of $\mathcal{O}\left(\frac{r^5 2^{8r}\alpha_A^2}{ ((\Delta')^{4r}\epsilon c_{\min})^2} \log\left(\frac{r}{\delta}\right)\right)$. 
    \item A total number of $\mathcal{O}\left(\frac{r^5 2^{8r}\alpha_A^2\alpha_\beta d^{3/2}\alphaFP ^2}{ ((\Delta')^{4r}\epsilon c_{\min})^2} \log\left(\frac{\alpha_A\alpha_\beta\sqrt{d}\alphaFP r^3 2^{4r}}{\epsilon c_{\min}(\Delta')^{4r}}\right)\log\left(\frac{\alpha_A r^3 2^{4r}}{\epsilon c_{\min}(\Delta')^{4r}}\right)\log\left(\frac{r}{\delta}\right)\right)$ queries to controlled $U_A$, $U_s$ and their inverse.
\end{itemize}
\end{theorem}
\begin{proof}
Similar to the real case, we set the contribution in operator norm of the noisy matrix in both the approximation of the matrix exponential and the shot noises to be $\varepsilon_1/2$.

We determine $\varepsilon_1$. Denote ${Q}=e^{\lambda_j/\alpha_A}$ for some $j$.  Suppose the GEVP solver outputs an estimate $\widetilde{Q}$ of $Q$ with error bound $\varepsilon_1$, i.e.,
\begin{equation*}
|Q - \widetilde{Q}| \leq \varepsilon_1\leq\frac{1}{2e},
\end{equation*}
where we have taken the error to be small (less than $\frac{1}{2e}$).
The estimator of the eigenvalue becomes $\widetilde{\lambda}_j:=\alpha_A \ln(\widetilde{Q})$.
We apply results from complex analysis to develop a mean-value-theorem-type bound. By \cite[Chapter 4, Section 1.1]{ahlfors1953complex}, for holomorphic function $f(x)$, $|f(a)-f(b)|=\left|\int_{a}^b f'(x) \mathrm{d}x\right|\leq \int_{a}^b |f'(x)| \mathrm{d}x$ such that the integral is a complex line integral. Hence, we can bound $\int_{a}^b |f'(x)| \mathrm{d}x\leq |a-b|\sup_{z \in [a,b]}|f'(z)|$ with $[a, b]$ denotes line segment connecting $a$ and $b$ in the complex plain. By taking $f(x)=\ln(x)$, we have that
\begin{align*}
    \bigl|\ln\widetilde Q - \ln Q\bigr|
&\leq  |\widetilde{Q}-Q| \max _{x \in[Q, \widetilde{Q}]}\left|\frac{1}{x}\right|\\
&\leq \frac{\bigl|\widetilde Q - Q\bigr|}{\min (|Q|,|\widetilde{Q}|)-|\widetilde{Q}-Q|}\\
&=\frac{\bigl|\widetilde Q - Q\bigr|}{\min (|Q|,|\widetilde{Q}|)-\varepsilon_1} \leq 2e \varepsilon_1.
\end{align*}
Here, in the second line, we have used $\max _{x \in[Q, \widetilde{Q}]}\left|\frac{1}{x}\right|\leq\frac{1}{\min (|Q|,|\widetilde{Q}|)-|\widetilde{Q}-Q|}$. In the last line, we have used that $|Q|=|e^{\lambda_j/\alpha_A}|\geq e^{-\alpha_A/\alpha_A}= e^{-1}$ so as $|\widetilde{Q}|$.
It then implies that 
$$
\bigl|\widetilde\lambda_j - \lambda_j\bigr|
=\alpha_A\,\bigl|\ln\widetilde Q - \ln Q\bigr|
\le 2e\,\alpha_A\,\varepsilon_1.
$$
To guarantee $|\widetilde\lambda_j-\lambda_j|\le\epsilon$, it suffices to choose
\begin{equation}\label{eq:epsilon1_temp}
    \varepsilon_1 =\frac{\epsilon}{2e\,\alpha_A}.
\end{equation}

The truncation order follows from \cite[Theorem 27]{low2024quantum}, which gives
$$\left\|e^{(A/\alpha_A)t}-\sum_{j=0}^{d-1} \beta_j {F}_j(A/\alpha_A)\right\|=\mathcal{O}\left(\left(\frac{e^\zeta t}{d}\right)^d\right),$$
where $F_j$ the degree $j$ Faber polynomial. As the maximal evolution time $t=\mathcal{O}(r)$, we have the exponential function approximated by a degree-($d-1$) Faber polynomial with $d=\mathcal{O}(e^\zeta r+\log(1/\varepsilon_2))$. And by Lemma \ref{lem:qevt_c_block}, the maximal query complexity is ${\mathcal{O}\left(\frac{\alpha_\beta\sqrt{d}\alphaFP}{\alpha_p}d\alphaFP\log\left(\frac{\alpha_\beta\sqrt{d}\alphaFP}{\alpha_p\varepsilon_2}\right)\log\left(\frac{1}{\varepsilon_2}\right)\right)}$. Remark that we can always choose $\alpha_p=\mathcal{O}(1)$. 

To determine $\varepsilon_2$, we want the operator norm of the noise matrix induced by the truncation to be smaller than $\frac{\varepsilon_1}{2}=\frac{\pi \epsilon}{4e\alpha_A}$. Then, similar to the real case, the operator norm of the noise matrix is bounded by the Frobenius norm, which gives $r\varepsilon_2$. We then apply Lemma \ref{lem:sample_complexity}, the operator norm of the truncation-induced noise matrix follows:
$$
\frac{r\varepsilon_2}{c_{\min}}\left( \frac{r^2 2^{4r}}{\pi^2\Delta'^{4(r-1)}}+\frac{r 2^{3r}}{\pi^{3/2}\Delta'^{3(r-1)}} \right) \leq \frac{\pi \epsilon}{4e\alpha_A}.
$$
We obtain 
\begin{equation}\label{eq:ep2}
    \varepsilon_2=\Theta\left(\frac{\epsilon c_{\min}(\Delta')^{4r}}{r^3 2^{4r}\alpha_A}\right).
\end{equation}

The sample complexity is addressed by considering the accuracy Eq.~\eqref{eq:epsilon1_temp} and applying it to Lemma \ref{lem:final_sample_complexity}.



\end{proof}

For pure initial states and when the purified quantum query access is available, we have:
\begin{theorem}[Complex eigenvalues with purified quantum query access]\label{thm:complex_with}
For $A$ with complex eigenvalues with pure or purified initial states, the costs become:
\begin{itemize}
    \item A maximal ${\mathcal{O}\left(\frac{r^3 2^{4r}\alpha_A\alpha_\beta d^{3/2}\alphaFP ^2}{\epsilon c_{\min}(\Delta')^{4r}}\log\left(\frac{\alpha_A\alpha_\beta\sqrt{d}\alphaFP r^3 2^{4r}}{\epsilon c_{\min}(\Delta')^{4r}}\right)\log\left(\frac{\alpha_A r^3 2^{4r}}{\epsilon c_{\min}(\Delta')^{4r}}\right)\right)}$ queries to (controlled) $U_A$ and its inverse in one coherent run of the algorithm.
    \item A sample complexity of $\mathcal{O}(r)$.
    \item A total number of ${\mathcal{O}\left(\frac{r^4 2^{4r}\alpha_A\alpha_\beta d^{3/2}\alphaFP ^2}{\epsilon c_{\min}(\Delta')^{4r}}\log\left(\frac{\alpha_A\alpha_\beta\sqrt{d}\alphaFP r^3 2^{4r}}{\epsilon c_{\min}(\Delta')^{4r}}\right)\log\left(\frac{\alpha_A r^3 2^{4r}}{\epsilon c_{\min}(\Delta')^{4r}}\right)\log\left(\frac{r}{\delta}\right)\right)}$ queries to controlled-$U_A$ and its inverse.
\end{itemize}
\end{theorem}
\begin{proof}
The proof is similar to Theorem \ref{thm:complex_without} with the only change that the error given by Eq.~\eqref{eq:ep2} is taken to the QAE algorithm.
\end{proof}

\section{Applications}\label{sec:app}
\subsection{Liouvillian gap estimation}
Estimating the Liouvillian gap $g_{\mathcal{L}}$ is critical because it quantifies the slowest decay mode of open quantum systems, with the inverse gap giving the asymptotic relaxation time $\tau\approx 1/g_{\mathcal{L}}$ toward steady state~\cite{mori2020resolving}. In dissipative quantum state engineering, a larger gap ensures rapid convergence to target entangled or resource states, improving preparation fidelity~\cite{verstraete2009quantum}. Conversely, in passive quantum memories and error‐correction schemes, a small Liouvillian gap leads to long‐lived noise modes that limit logical qubit lifetimes by allowing environmental errors to accumulate~\cite{labay2025chiral}. In quantum heat engines and refrigerators, the gap controls the working medium’s energy exchange rate with thermal baths, thus setting upper bounds on cycle frequency and power output~\cite{zhang2022dynamical}. For quantum algorithms modeled as Markov processes, such as Gibbs sampling or optimization, mixing times scale inversely with the Liouvillian gap, directly informing runtime and convergence guarantees~\cite{ramkumar2024mixing}. Non‐Hermitian skin effects can drastically alter spectral properties and lead to boundary‐dependent gap sizes, resulting in localization‐induced relaxation anomalies. Introducing a symmetrized Liouvillian gap refines transient‐dynamics bounds by bounding initial decay before asymptotic rates dominate. In many‐body systems under weak dissipation, singular behavior of the gap in the thermodynamic limit gives rise to long‐lived quasi‐steady states that challenge control protocols. By employing Floquet and exceptional‐point engineering, one can tune and widen the Liouvillian gap to accelerate or decelerate relaxation as desired. Accurate gap estimates are therefore indispensable for designing real‐time feedback and reservoir‐engineering schemes, since they set the required rates of corrective operations to stabilize quantum states.

The dynamics of an open quantum system is described by the Lindblad master equation $\frac{{\rm d}\rho}{{\rm d}t}=\mathcal{L}(\rho)$, where $\mathcal{L}:\mathbb{C}^{2^n\times 2^n}\mapsto \mathbb{C}^{2^n\times 2^n}$ represents a superoperator. More specifically, the Lindblad operator
\begin{align}
    \mathcal{L}(\rho)=-i[H,\rho]+\sum\limits_{\mu}\left(L_{\mu}\rho L^{\dagger}_{\mu}-\frac{1}{2}\{L^{\dagger}_{\mu}L_{\mu},\rho\}\right),
\end{align}
governed by some Hamiltonian $H$ and dissipative operator $L_{\mu}$. To utilize the proposed quantum algorithm, the Lindblad operator $\mathcal{L}$ is required to ``vectorized" to a matrix. Suppose the underlying Hamiltonian
\begin{align}
    H=\sum\limits_{j}\alpha_j V_{0,j},
\end{align}
and the dissipative term
\begin{align}
    L_{\mu}=\sum\limits_{j}\sqrt{\alpha_{\mu,j}}V_{\mu,j}
\end{align}
with $\alpha_j$, $\alpha_{\mu, j}$ are real values, and $V_{\mu,j}$ are $n$-qubit Pauli operators. By performing the vectorization, the Lindblad equation can be expressed via $d|\rho\rangle\rangle/dt=\tilde{\mathcal{L}}|\rho\rangle\rangle$, where the density matrix 
\begin{align}
    |\rho\rangle\rangle=\sum\limits_{j,k}\rho_{j,k}|j\rangle|k\rangle,
\end{align}
with the entry $\rho=\langle j|\rho|k\rangle$. Meanwhile the superoperator
\begin{eqnarray}
\begin{split}
     &\tilde{\mathcal{L}}=\sum\limits_{\mu}\tilde{\mathcal{L}}_{\mu},\\
     &\tilde{\mathcal{L}}_0=-iH\otimes I^{\otimes n}+iI^{\otimes n}\otimes H,\\
     &\tilde{\mathcal{L}}_{\mu\geq 1}=\left(L_{\mu}\otimes L_{\mu}^*-\frac{1}{2}L_{\mu}^{\dagger}L_{\mu}\otimes I^{\otimes n}-\frac{1}{2}I^{\otimes n}L_{\mu}^TL_{\mu}^*\right).
\end{split}
\end{eqnarray}
Let $\lambda_j(\tilde{\mathcal{L}})$ be the eigenvalues of $\tilde{\mathcal{L}}$ such that $0={\rm Re}[\lambda_0(\tilde{\mathcal{L}})]>{\rm Re}[\lambda_1(\tilde{\mathcal{L}})]>\cdots$, and the Liouvillian gap is defined as $g_{\mathcal{L}}=\abs{{\rm Re}[\lambda_1(\tilde{\mathcal{L}})]}$.

Given the initial density matrix as $\rho_I=|\rho\rangle\rangle \langle \langle\rho|$, we have 
$$\tr{L\rho_I}=\sum_{i=1}^r \lambda_i \tr{\ket{\psi_i}\bra{\phi_i}\rho_I}=\sum_{i=1}^r c_i \lambda_i,$$
where $L$ is the vectorized Liouvillian with eigen-decomposition $L=\sum_{i} \lambda_i \ket{\psi_i}\bra{\phi_i}$, and we have assumed the $r$-sparse condition in Assumption \ref{assm:sparsity}. Besides, we need the sparse decomposition to contain the eigen-component corresponding to $\lambda_1$. This is a reasonable assumption, such that one can obtain such an initial state through constructions used in linear response theory. That is, we first prepare the steady state (eigenstate corresponding to $\lambda_0=0$), and we may `excite' the state to a linear combination of eigen-components of low-lying states.

For estimating the Liouvillian gap, we can adopt the construction given by Theorem \ref{thm:general_with}. As long as $c_{\min}$ and the normalized gap $\Delta'$ remain reasonable, we can estimate the $r$ eigenvalues, including $\lambda_1$, to desirable accuracy with Heisenberg scaling guarantees. This gives us the Liouvillian gap.

\subsection{Spectral abscissa and dynamics stability}
The spectral abscissa of a matrix is a fundamental quantity in the stability analysis of dynamical systems, particularly in fields like control theory, fluid dynamics, and differential equations. In fluid dynamics, non-normal matrices often arise when linearizing the Navier-Stokes equations around steady states, such as laminar flows. The spectral properties—especially the dominant eigenvalue—inform the stability of these flows and help explain phenomena like transition to turbulence.

In such contexts, the non-normal linearized Navier-Stokes operator (matrix) $\mathcal{N}$  governs the evolution of infinitesimal perturbations $\ket{u(t)}$ in fluid flows:
\begin{equation}
\frac{\partial \ket{u(t)}}{\partial t} = \mathcal{N} \ket{u(t)}.
\end{equation}
The spectral abscissa is defined as
$$a(\mathcal{N}):=\max\{\mathrm{Re}(\lambda):\lambda\in\sigma(\mathcal{N})\},$$
where $\sigma(\mathcal{N})$ denotes the set of eigenvalues of $\mathcal{N}$. The spectral abscissa's importance lies in: 
The sign of the spectral abscissa determines the stability of the dynamics:
\begin{itemize}
    \item If $a(\mathcal{N})<0$, all trajectories decay exponentially to zero, implying asymptotic stability.
    \item If $a(\mathcal{N})=0$, the system is linearly marginally stable. Yet, it may exhibit transient growth due to non-normality.
    \item If $a(\mathcal{N})>0$, the trajectories grow exponentially, meaning instability.
\end{itemize}

Similar to the discussion in the last section, we assume an initial density matrix $\rho_I$ that contains eigen-components corresponding to the spectral abscissa non-trivially (the expanded coefficient is not exponentially small). We can first estimate $a(\mathcal{N})$ using methods discussed in Sec.~\ref{sec:general_solver}. If we find $a(\mathcal{N})>0$, we may directly output the estimation. Yet, this method cannot detect the zero eigenvalue. Thus, if we otherwise find that the estimated value is less than zero, we further adopt the method in Sec.~\ref{sec:complex_solver} to determine whether there exists an eigenvalue of value zero.

\subsection{Non-normal matrices with real eigenvalues}
In this section, we identify cases where we know a non-normal matrix has real eigenvalues. For such matrices, we can apply the method in Sec.~\ref{sec:real_solver} for solving their eigenvalues.

In non-Hermitian quantum systems, Hamiltonians that are PT-symmetric can exhibit entirely real spectra~\cite{bender1998real}. The PT symmetry dictates that the system is invariant under combined parity and time-reversal operations.

Totally nonnegative matrices~\cite{fallat2011totally} represent a fundamental class of matrices where all minors are nonnegative. It is established that totally nonnegative matrices possess entirely real eigenvalues~\cite{fallat2011totally}, making them particularly valuable for applications requiring spectral reality guarantees. These matrices frequently exhibit non-normal characteristics while maintaining their real spectral properties. Totally nonnegative matrices arise in specialized contexts, including certain population dynamics models with specific age-structured constraints~\cite{logofet2008nonnegative}, discretized Green's functions for boundary value problems~\cite{stakgold2011green}, particular classes of interpolation matrices in numerical analysis~\cite{fallat2011totally}, and oscillatory matrices in mechanical vibration analysis~\cite{gantmacher2002oscillation}.

Upper and lower triangular matrices with real diagonal entries have real eigenvalues, while being non-normal. That is, all eigenvalues equal these diagonal elements and are therefore real. These matrices are generally non-normal unless they are diagonal.

\end{document}